\documentclass[11pt,letterpaper]{article}

\usepackage{graphicx}
\usepackage[utf8]{inputenc}
\usepackage[margin=1in]{geometry}
\usepackage{amsmath,amsthm,amssymb}
\usepackage{bbm}
\usepackage{xcolor}
\usepackage{float}
\usepackage{natbib}
\usepackage{bbold}
\usepackage{mathtools}
\usepackage{doi}
\usepackage{tikz}
\usepackage{thm-restate}
\usepackage{algorithm, algpseudocode}
\usepackage{booktabs}

\usepackage{hyperref}
\usepackage{subcaption}
\usepackage{enumitem}

\usepackage[capitalize,noabbrev,nameinlink]{cleveref}
\usepackage{subcaption}
\usepackage{algorithm}
\usepackage{algpseudocode}
\usepackage{crossreftools}

\theoremstyle{plain}
\newtheorem{theorem}{Theorem}[section]
\newtheorem{proposition}[theorem]{Proposition}
\newtheorem{lemma}[theorem]{Lemma}

\theoremstyle{definition}
\newtheorem{definition}{Definition}[section]

\theoremstyle{remark}
\newtheorem{remark}{Remark}[section]

\makeatletter
\AddToHook{cmd/appendix/before}{\def\cref@section@alias{appendix}\def\cref@subsection@alias{appendix}}
\makeatother

\definecolor{DarkGreen}{rgb}{0.1,0.5,0.1}
\definecolor{DarkRed}{rgb}{0.5,0.1,0.1}
\definecolor{DarkBlue}{rgb}{0.1,0.1,0.5}
\definecolor{Gray}{rgb}{0.2,0.2,0.2}

\hypersetup{
    unicode=false,
    pdftoolbar=true,
    pdfmenubar=true,
    pdffitwindow=false,
    pdfnewwindow=true,
    colorlinks=true,
    linkcolor=DarkBlue,
    citecolor=DarkGreen,
    filecolor=DarkRed,
    urlcolor=DarkBlue,
    pdftitle={},
    pdfauthor={},
}

\makeatletter
\def\maketag@@@#1{\hbox{\m@th\normalfont\normalsize#1}}
\makeatother

\title{Lower Bias, Higher Welfare: How Creator Competition Reshapes Bias-Variance Tradeoff in Recommendation Platforms?}

\author{
Kang Wang\thanks{Institute for Theoretical Computer Science, Shanghai University of Finance and Economics. Email: \texttt{wangkang0330@gmail.com}}
\and
Renzhe Xu\thanks{Corresponding author. MoE Key Lab of Interdisciplinary Research of Computation and Economics. Institute for Theoretical Computer Science, Shanghai University of Finance and Economics. Email: \texttt{xurenzhe@sufe.edu.cn}}
\and
Bo Li\thanks{School of Economics and Management, Tsinghua University. Email: \texttt{libo@sem.tsinghua.edu.cn}}
}

\date{}

\newcommand{\calC}{\mathcal{C}}
\newcommand{\calG}{\mathcal{G}}

\newcommand{\calN}{\mathcal{N}}
\newcommand{\bbP}{\mathbb{P}}
\newcommand{\bbR}{\mathbb{R}}
\newcommand{\bbRpos}{\bbR_{\geq 0}}
\newcommand{\bbS}{\mathbb{S}}
\newcommand{\bbSpos}{\bbS_{\geq 0}}
\newcommand{\bbI}{\mathbb{I}}
\newcommand{\bbE}{\mathbb{E}}
\newcommand{\bfI}{\mathbf{I}}
\newcommand{\bfR}{\mathbf{R}}
\newcommand{\bfr}{\mathbf{r}}
\newcommand{\bfs}{\mathbf{s}}
\newcommand{\bfu}{\mathbf{u}}
\newcommand{\hbfu}{\hat{\bfu}}
\newcommand{\bfv}{\mathbf{v}}
\newcommand{\bfeps}{\boldsymbol{\epsilon}}

\newcommand{\spne}{\bfs^{\text{PNE}}}
\newcommand{\hu}{\hat{u}}

\newcommand{\lambCom}{\lambda_\mathrm{str}^*}
\newcommand{\lambComU}{\lambda_\mathrm{str}^{*\mathrm{U}}}
\newcommand{\lambSta}{\lambda_\mathrm{non}^*}
\newcommand{\lambStaU}{\lambda_\mathrm{non}^{*\mathrm{U}}}
\newcommand{\lambStaL}{\lambda_\mathrm{non}^{*\mathrm{L}}}

\newcommand{\barepsilonT}{\bar\epsilon_T}
\newcommand{\barepsilonN}{\bar\epsilon_N}
\newcommand{\nT}{\calN_T}
\newcommand{\nN}{\calN_N}

\newcommand{\barE}{\bar E}
\newcommand{\cT}{\theta_T}
\newcommand{\cN}{\theta_N}
\newcommand{\TvN}{PreNT}
\newcommand{\BRCM}{BRCM^*}
\newcommand{\Mexpo}{M^3(\mathrm{expo.})}
\newcommand{\Menga}{M^3(\mathrm{enga.})}
\newcommand{\Mzero}{M^3(\mathrm{0})}

\DeclareMathOperator*{\argmax}{arg\,max}

\begin{document}
\maketitle

\begin{abstract}
Understanding the bias-variance tradeoff in user representation learning is essential for improving recommendation quality in modern content platforms. While well studied in static settings, this tradeoff becomes significantly more complex when content creators strategically adapt to platform incentives. To analyze how such competition reshapes the tradeoff for maximizing user welfare, we introduce the \textbf{Content Creator Competition with Bias-Variance Tradeoff} ($C^3_{BV}$) framework, a tractable game-theoretic model that captures the platform’s decision on regularization strength in user feature estimation. We derive and compare the platform’s optimal policy under two key settings: a non-strategic baseline with fixed content and a strategic environment where creators compete in response to the platform's algorithmic design.

Our theoretical analysis in a stylized model shows that, compared to the non-strategic environment, content creator competition shifts the platform’s optimal policy toward weaker regularization, thereby favoring lower bias in the bias-variance tradeoff. To validate and assess the robustness of these insights beyond the stylized setting, we conduct extensive experiments on both synthetic and real-world benchmark datasets. The empirical results consistently support our theoretical conclusion: in strategic environments, reducing bias leads to higher user welfare. These findings offer practical implications for the design of real-world recommendation algorithms in the presence of content creator competition.


\end{abstract}

\section{Introduction}


Machine learning models, particularly recommendation algorithms, have become integral to online platforms, profoundly shaping user experiences and the digital content landscape. At the core of such systems lies a fundamental design principle: the \emph{bias-variance tradeoff} \citep{shivaswamy2021bvranking,belkin2019reconciling,yang2020rethinking}. For example, when estimating user features, platforms must carefully balance this tradeoff---typically through regularization---to derive meaningful representations of user preferences from user-content interaction data, which are inherently sparse and noisy. A high-bias model (strong regularization) may offer stability but fail to capture nuanced user tastes, while a low-bias model (weak regularization) can provide greater personalization at the risk of overfitting to noise. The platform's choice of where to strike this balance is a critical lever that directly impacts the quality of its recommendations and, consequently, overall social and economic impact.

Classical recommendation frameworks often assume a \textit{separation} between content creation and recommendation algorithm design, modeling content generation as independent of the recommenders \citep{nguyen2013content,raza2024comprehensive}. Under this paradigm, content features are typically treated as \textit{static} inputs. This traditional perspective, however, overlooks a crucial element of modern recommender ecosystems: the \textit{strategic} behavior of content creators \citep{jagadeesan2025accounting}. Specifically, creators are not passive data points; they are rational agents who actively adapt their strategy (i.e., the content they produce) in response to the platform's algorithm to maximize their visibility and utility. As a result, a dynamic feedback loop emerges in which the platform’s design shapes creator strategies and, in turn, alters the content landscape. This leads to a critical but underexplored question: how does the strategic nature of content creation reshape the platform’s optimal bias-variance tradeoff?

To answer this question, we propose the \textbf{Content Creator Competition with Bias-Variance Tradeoff} ($C^3_{BV}$), a framework that builds on the foundational work of \citet{yao2023howbad, yao2023rethinking} to jointly model the bias-variance tradeoff in user estimation and the strategic competition among content creators. At the core of our model is the platform's choice of the regularization coefficient (i.e., $\lambda$) in learning user preferences, which serves as the primary policy lever to control this tradeoff. Following prior work \citep{yao2023rethinking, yao2024user}, we assume the platform aims to maximize total user welfare, measured by the scale and quality of user-content interactions. This objective aligns closely with the platform's revenue~\citep{yao2024user} and reflects the system’s broader social value~\citep{rosenfeld2025machine}. To investigate the impact of creator competition, we analyze the platform’s optimal choice of $\lambda$, denoted by $\lambSta$ and $\lambCom$, under two distinct environments: a \textit{non-strategic} baseline, where content remains fixed, and a \textit{strategic} setting, where content creators adapt their strategies in response to the platform-provided incentives. Our objective is to characterize how content competition drives the shift from $\lambSta$ to $\lambCom$, thereby reshaping the platform’s optimal bias-variance tradeoff.

Although several prior works have examined the bias-variance tradeoff in competitive settings \citep{feng2022bias, iyer2024competitive}, our approach fundamentally differs in \emph{who} determines the tradeoff. Specifically, existing studies typically adopt a \emph{market participant} perspective, where each firm individually selects its own tradeoff strategy. For instance, \citet{iyer2024competitive} analyze how competing firms choose bias-variance parameters in targeting. In contrast, our framework adopts a \emph{market designer} perspective: the platform, rather than individual creators, selects the regularization parameter when estimating user features. Content creators then compete strategically based on these platform-determined estimates. This distinction better reflects the architecture of modern recommender systems, where the platform centrally controls user modeling. As a result, prior analyses do not directly apply to our setting.

We begin with a stylized setting of the $C^3_{BV}$ framework to enable tractable theoretical analysis. Specifically, we consider a single-user scenario and assume that the platform’s content can be categorized into two types: trend content and niche content. This setting reflects real-world platforms, where some content attracts widespread attention while other content remains relatively unpopular~\citep{gill2007youtube, salganik2006experimental}. Our analysis shows that the relationship between $\lambSta$ and $\lambCom$ depends critically on the user’s content preference. When the user favors trend content, $\lambSta$ is bounded below by a positive constant, while $\lambCom$ is bounded above, and under mild conditions, we prove that $\lambCom < \lambSta$. In contrast, when the user prefers niche content, both $\lambSta$ and $\lambCom$ are upper bounded and typically close in value, making them difficult to distinguish. These findings suggest that for trend-preferring users, the platform should apply weaker regularization (i.e., reduce bias) in the strategic environment, while for niche-preferring users, regularization should remain small in both cases. Given that aggregate user preferences in real-world platforms often skew toward trend content \citep{xu2023smaa, fleder2009blockbuster}, we conclude that the optimal regularization in the \textit{strategic} environment ($\lambCom$) is generally smaller than in the \textit{non-strategic} environment ($\lambSta$), implying that the platform should favor lower bias when accounting for strategic content competition.

We further validate our theoretical findings through experiments on both synthetic and real-world datasets. In the synthetic setting, we extend the stylized single-user model to multi-user environments and construct two scenarios: a Trend Market and a Niche Market, representing user populations with predominantly trend-preferring and niche-preferring preferences, respectively. The results are consistent with theory: in the Trend Market, we observe $\lambCom < \lambSta$, whereas in the Niche Market, both values remain small and close. To assess the practical relevance of these insights, we conduct extensive experiments on two widely used benchmark datasets, MovieLens-100k and Amazon Product Reviews, under a variety of recommendation algorithm configurations. The results consistently show that $\lambCom < \lambSta$ across all settings, suggesting that the platform should adopt weaker regularization in the presence of content creator competition. This corresponds to favoring lower bias at the cost of higher variance, which can better serve the majority of users. These findings yield important implications for the design and implementation of recommendation algorithms in environments with competing content creators.

To conclude, our paper provides the following key takeaways for industry practitioners:
\begin{enumerate}
    \item We propose the Content Creator Competition with Bias-Variance Tradeoff ($C^3_{BV}$) framework, which models the bias-variance tradeoff in user feature estimation in the presence of strategic content creator competition.
    \item Through a stylized setup, we theoretically analyze this tradeoff and show that platforms should prefer lower bias when content creators behave strategically.
    \item Extensive experiments on both synthetic and real-world datasets support our theoretical findings, consistently indicating that the platform should reduce bias in competitive content environments.
\end{enumerate}

\section{Related Works}







\paragraph{Bias-variance tradeoff}
The bias–variance tradeoff is a fundamental concept in statistical learning, describing the inverse relationship between bias—systematic errors due to model assumptions---and variance---fluctuations arising from sensitivity to training data \citep{belkin2019reconciling,yang2020rethinking,feng2022bias}. This tradeoff is commonly managed through regularization techniques, which balance model complexity to reduce total prediction error \citep{iyer2024competitive, mnih2007pmf}. In recommender systems, the bias–variance tradeoff usually manifests in the tension between accurately modeling user preferences (reducing bias) and maintaining generalizability across sparse and noisy interaction data (controlling variance) \citep{shivaswamy2021bvranking,kawakami2021investigating}. Collaborative filtering methods, for instance, may suffer from high variance when user–item matrices are sparse \citep{paterek2007improving,hu2008collaborative,li2019cfnoisy,he2016ups}, while content-based methods may incur high bias due to limited feature expressiveness \citep{eilat2023performative,shu2018content,yoshitaka1999survey}.

While traditional analyses of the bias–variance tradeoff assume static data, recent studies have begun to examine this tradeoff in competitive environments involving interacting agents \citep{feng2022bias, iyer2024competitive}. These approaches typically treat the bias–variance balance as a strategic choice made by individual players. In contrast, our work adopts a \textit{market designer} perspective, framing the tradeoff not as an agent-level decision but as a policy instrument employed by the platform to shape ecosystem dynamics and enhance aggregate user welfare.


\paragraph{Content recommendation and online platforms}
Modern content recommendation systems are widely deployed by online platforms such as Netflix, YouTube, and TikTok, where algorithms mediate interactions between users and content creators. Traditional approaches---including collaborative filtering methods \citep{hu2008collaborative, koren2009matrix} and deep learning models \citep{li2019cfnoisy, covington2016youtube}---have significantly improved recommendation accuracy, yet often treat both users and creators as passive data sources. In contrast, a growing body of research highlights the active role of recommendation algorithms in shaping user preferences and creator behaviors via feedback loops, modeling content creators as strategic agents who adapt their behavior in response to platform incentives \citep{jagadeesan2025accounting}. Recent studies have explored algorithmically curated recommenders \citep{hron2023modeling}, strategic competition among creators and resulting equilibria \citep{jagadeesan2023supply}, as well as the broader social and performative implications \citep{yao2023howbad, yao2023rethinking, yu2025beyond, eilat2023performative}.

Distinct from prior work, our $C^3_{BV}$ framework bridges the gap between machine learning model design and downstream societal consequences, focusing on the influence of recommendation algorithms as exercised by platforms.

\paragraph{Other competition frameworks}

Finally, our work also relates to broader research on competition in other machine learning settings \citep{jagadeesan2023improved, ben2019regression, xu2025heterogeneous}, targeted advertising \citep{iyer2024competitive, iyer2024precision}, market participation \citep{einav2025market, spulber2019economics}, and strategic behavior modeling \citep{hardt2016strategic, perdomo2020performative}. Unlike these studies, we examine competition outcomes that arise specifically from the bias-variance tradeoff in recommendation platforms.

\section{Model}




\subsection{Preliminaries}





We begin by introducing the bias-variance tradeoff that arises in the estimation of user feature representations.

Consider a recommender system with $m$ users and $n$ pieces of content. Let $u_i \in \mathbb{R}^d$ and $v_j \in \mathbb{R}^d$ denote the latent feature vectors of user $i$ and content $j$, respectively. Define $\bfu = (u_1, \ldots, u_m)$ and $\bfv = (v_1, \ldots, v_n)$. If these latent features were known, the platform could directly compute the most relevant content for each user and make personalized recommendations. In practice, however, the platform does not have access to the ground-truth features $\bfu$ and $\bfv$ and must instead estimate them from observed data.

A widely adopted approach for the feature estimation in both academia~\citep{dean2020reachability, srilakshmi2022two-stage, hu2008collaborative, greenstein2017session} and industry~\citep{gomez2015netflix, covington2016youtube} is a two-stage method: one first estimates either user or content features, and then estimates the other based on user-content interaction data. For example, YouTube incorporates rich content features to construct user embeddings within its recommendation pipeline~\citep{covington2016youtube}, while Netflix employs video embeddings from its ``sims'' algorithms to power its \textit{Because You Watched} recommendations~\citep{gomez2015netflix}. In practice, content feature extraction has become significantly more effective and accessible due to recent advances in deep learning for language~\citep{achiam2023gpt, liu2024deepseek}, vision~\citep{dosovitskiy2020image}, and multi-modal modeling~\citep{yin2024survey}. As a result, following \citet{agarwal2024system}, we assume throughout this paper that the platform has access to ground-truth content embeddings and estimates user features based on these embeddings and the observed user-content interaction data.

These user-content interactions are represented by an observation matrix $\bfR \in \mathbb{R}^{m \times n}$, where each entry $R_{ij}$ denotes the interaction (e.g., a rating or a click) from user $i$ with content $j$. For simplicity, we assume that all entries in $\bfR$ are observed, although our results can be readily extended to settings with partially observed interactions. From a probabilistic perspective~\citep{mnih2007pmf, li2019cfnoisy}, each interaction can be modeled as the sum of a latent utility and observation noise:
\begin{equation} \label{eq:R-equals-US-plus-noise}
    R_{ij} = u_i^\top v_j + \epsilon_{ij},
\end{equation}
where $\epsilon_{ij}$ denotes i.i.d. random noise with zero mean. As a result, the feature of user $i \in [m]$ can then be estimated by minimizing the following regularized loss function:
\begin{equation} \label{eq:mf-loss-function}
    \mathcal{L}_i(u) = \sum_{j=1}^n (R_{ij} - u^\top v_j)^2 + \lambda \|u\|_2^2.
\end{equation}
Minimizing the objective in \cref{eq:mf-loss-function} via the least squares method yields the following closed-form solution:
\begin{equation} \label{eq:mf-hat-u}
    \hu_i(\lambda) = \left(\sum_{j=1}^n v_jv_j^\top + \lambda \bfI\right)^{-1}\left(\sum_{j=1}^n R_{ij}v_j\right).
\end{equation}
Here $\hu_i(\lambda)$ represents the estimated feature vector for user $i$. Define $\hbfu(\lambda) = (\hu_1(\lambda), \hu_2(\lambda), \dots, \hu_m(\lambda))$. For notational simplicity, we use $\hbfu$ and $\hu_i$ to denote the estimated quantities when the dependence on $\lambda$ is clear from context.

\subsection{Content Creator Competition with Bias-Variance Tradeoff}

In this section, we introduce the Content Creator Competition with Bias-Variance Tradeoff ($C^3_{BV}$), to study the related properties and outcome of a content creator competition in a real-world machine learning scenario, where user features are estimated under noisy observations. Here a $C^3_{BV}$ game instance $\calG$ can be described by a tuple $(\lambda, \bfu, \bfv, \bfeps, K, \sigma, RM, \bfr)$, illustrated as follows:

\paragraph{Basic setups} We consider an online platform with $m \geq 1$ users and $n \geq 2$ pieces of content, and begin our setup from the perspective of user and content embeddings. $\bfu = (u_1, \ldots, u_m)$ represents the ground-truth latent features of users. $\bfv = (v_1, \ldots, v_n)$ represents the latent embeddings of the content items initially available on the platform. The matrix $\bfeps = \{\epsilon_{ij}\}_{i \in [m],\, j \in [n]}$ represents the entry-wise i.i.d. noise in the observed rating matrix described in \cref{eq:R-equals-US-plus-noise}. Based on this noisy interaction data and the known content embeddings, the platform estimates user features using the closed-form solution in \cref{eq:mf-hat-u}, resulting in the estimated user embeddings $\hbfu = (\hu_1, \ldots, \hu_m)$. Motivated by prior work~\citep{hron2023modeling, eilat2023performative}, we assume that all user embeddings lie in the nonnegative orthant, i.e., $u_i, \hu_i \in \bbRpos^d$, and that content embeddings are constrained to the unit sphere in $\bbRpos^d$, i.e., $v_j \in \bbSpos^{d-1} \subset \bbRpos^d$.

In our setting, each content creator $j$ can modify their generated content from the initial embedding $v_j$ to a new embedding $s_j \in \bbSpos^{d-1}$. The vector $s_j$ is thus interpreted as the \textit{pure strategy} of creator~$j$. A strategy profile, representing the collection of strategies chosen by all creators, is denoted by the tuple $\bfs = (s_1, \dots, s_n)$.

\paragraph{Matching function and estimated matching score}
The platform's recommendation algorithm matches content to users via a matching function, $\sigma(s, u) : \bbR^d \times \bbR^d \to \bbR_{\geq 0}$, which quantifies the \textit{similarity} between a content feature $s$ and a user feature $u$.
While various forms of $\sigma$ have been explored, ranging from inner products \citep{jagadeesan2023supply, rendle2020mfrevisited} to deep neural networks \citep{zamani2020learning, li2021path}, we adopt the inner product as the default matching function for all subsequent analysis.
For a given user $i$, the platform uses their estimated feature (i.e., $\hu_i$), to compute scores with all $m$ content. We refer to these as the \textbf{estimated matching scores}, where the score for content $j$ is defined by the shorthand $\sigma_j(\hu_i) \equiv \sigma(s_j, \hu_i)$.


\paragraph{Rewarding mechanism and creator utility} 
Following prior work~\citep{yao2023howbad,yao2023rethinking}, we assume that each creator derives her utility from the rewards provided by the platform through a certain rewarding mechanism. Specifically, given a strategy profile $\bfs$, the utility of creator $j$ is the expected reward aggregated over all $n$ users:
\begin{equation} \label{eq:creator_j-utility}
    \pi_j(\bfs) = \sum_{i=1}^mRM(\sigma_j(\hu_i), \sigma_{-j}(\hu_i)),
\end{equation}
where $\sigma_{-j}(\hu_i)$ is the set of estimated matching scores between user $i$ and all creators other than $j$ and $RM: \bbR^n \rightarrow \bbR$ denotes the rewarding mechanism. We assume that $RM(\cdot,\cdot)$ satisfies \textit{individual monotonicity}. This property requires that a creator's reward does not decrease when their estimated matching score with a user increases, holding all other scores constant. Formally, for any creator~$t$, given a score vector $(\sigma_1, \dots, \sigma_n)$ and a deviation $\sigma'_t > \sigma_t$, we have:
\[
    RM(\sigma'_t, \sigma_{-t}) \geq RM(\sigma_t, \sigma_{-t}).
\]
This individual monotonicity property provides a crucial incentive guarantee: a creator is never penalized (suffer from utility loss) for improving their content's alignment with a user embedding  (specifically, the \textit{predicted} embedding in our framework). It is therefore a natural and desirable criterion for designing rewarding mechanisms and is satisfied by many widely-used examples.


\paragraph{User utility and user welfare} 
When user $i$ receives a recommended item $s_j$, she derives utility based on her true preference, i.e., $\sigma(s_j, u_i)$, rather than the estimated matching score $\sigma_j(\hu_i)$. Consequently, given the creators' strategy profile $\bfs$, the utility gained by user $i$ is:
\begin{equation} \label{eq:user_i-utility}
    W_i(\bfs, u_i, \hu_i) = \sum_{k=1}^{K} r_k \cdot \sigma(s_{j^{(k)}}, u_i),
\end{equation}
where $s_{j^{(k)}} = \arg\max^{(k)}_{s_j \in s} \sigma(s_j, \hu_i)$ denotes the $k$-th ranked content according to the estimated matching scores $\{\sigma_j(\hu_i)\}_{j=1}^{n}$, sorted in descending order. The sequence $\bfr = \{r_k\}_{k \in [K]}$ denotes the set of the user's attention weights on the Top-$K$ recommended slots, where $r_1 \geq r_2 \geq \dots \geq r_K \geq 0$.


The total user welfare is defined as the total utility all users gain, given any strategy profile $\bfs$ at the game instance $\calG$, i.e.,
\begin{equation} \label{eq:user-welfare}
   W(\bfs; \calG) = \sum_{i=1}^{n}W_i(\bfs, u_i, \hu_i).
\end{equation}

\paragraph{Equilibrium of the game}
The pure Nash equilibrium (PNE) \citep{nash1950equilibrium} is the most commonly used concept for characterizing the outcome of a game. It describes a stable state where no player has an incentive to unilaterally deviate their strategy, assuming the other players' strategies remain unchanged. Here we denote the strategy profile under PNE in a game instance $\calG$ as $\spne(\calG)$.

\subsection{Non-Strategic vs. Strategic Platforms}
We investigate the platform's bias-variance trade-off by contrasting two typical platform paradigms. One provides a non-competitive baseline by assuming fixed item embeddings, like Netflix and Amazon. The other models a competitive environment by anticipating response to game-theoretic incentives from creators, like TikTok and YouTube. This comparison allows us to reveal how strategic competition alters the platform's intrinsic bias-variance preferences relative to the baseline.

We begin by defining the non-strategic platform, which serves as a key baseline for evaluating the platform’s policy by isolating its effects from the complexities introduced by content creator competition. While the full game instance is defined as $(\lambda, \bfu, \bfv, \bfeps, K, \sigma, RM, \bfr)$, we simplify the notation to $\calG(\lambda, \bfv, \bfeps)$ in this section, as the other components are fixed and implicit in the following definitions.

\begin{definition}[Non-strategic Platform] \label{def:static-platform}
    A \textbf{non-strategic platform} assumes that content embeddings are fixed and not subject to strategic manipulation by creators; that is, the strategy profile $\bfs$ is set to the initial content embeddings $\bfv$. Under this setting, the optimal regularization parameter $\lambda$ that maximizes the expected total user welfare is given by
    \begin{equation} \label{eq:lambsta-opti}
        \lambSta = \argmax_{\lambda \geq 0} \bbE_{\bfeps} \left[W(\bfv; \calG(\lambda, \bfv, \bfeps))\right],
    \end{equation}
    where the expectation is taken over the observation noise $\bfeps$.
\end{definition}

Then the strategic platform is defined as follows.
\begin{definition}[Strategic Platform] \label{def:competition-platform}
    A \textbf{strategic platform} models content creators as rational, utility-maximizing agents who competitively adapt their strategies in response to the platform's reward mechanism within a game instance $\calG(\lambda, \bfv, \bfeps)$. In this setting, the resulting strategy profile is given by the subgame-perfect Nash equilibrium $\spne(\calG(\lambda, \bfv, \bfeps))$. The platform’s optimal regularization parameter is then defined as
    \begin{equation} \label{eq:lambcom-opti}
        \lambCom = \argmax_{\lambda \geq 0} \bbE_{\bfeps} \left[W\left(\spne(\calG(\lambda, \bfv, \bfeps)); \calG(\lambda, \bfv, \bfeps)\right)\right],
    \end{equation}
    where the expectation is taken over the observation noise $\bfeps$.
\end{definition}

Our objective is to analyze the relationship between $\lambSta$ and $\lambCom$ to characterize how content competition reshapes the platform’s optimal bias-variance tradeoff.



\section{Theories} \label{sec:theories}


In this section, we present the theoretical analysis based on a stylized instance of our $C^3_{BV}$ game, referred to as the $\TvN$ game.

\begin{definition}[$\TvN$ Game] \label{def:tvnEnvironment}
    The user \textbf{Pre}ference on \textbf{N}iche and \textbf{T}rend content ($\TvN$) Game is a stylized instance designed to model content creator competition on a platform with a \textit{single user} and an initial content pool consisting of either \textit{trend} or \textit{niche} items. The key components $\bfu$, $\bfv$, and $\bfeps$ in the game tuple are specified as follows:
    \begin{enumerate}[itemsep=0pt, topsep=2pt,left=0pt]
        \item \textbf{Initial Content Features $\bfv$:} The platform offers two content groups: a \textbf{trend group} with $\nT$ items and a \textbf{niche group} with $\nN$ items, where $\nT > \nN$ reflects the greater popularity of trend content. Each group shares a representative unit vector, $v_T$ for trend and $v_N$ for niche, which are orthogonal (i.e., $v_T^\top v_N = 0$, $\|v_T\|_2 = \|v_N\|_2 = 1$) to capture distinct content specializations~\citep{jagadeesan2023supply}. Then,
        \[
            \bfv = (\underbrace{v_T, \ldots, v_T}_{\nT \text{ times}}, \underbrace{v_N, \ldots, v_N}_{\nN \text{ times}}).
        \]
        
        \item \textbf{User Preference $\bfu$:} The user's ground-truth preference vector $u$ is modeled as a linear combination of the two content vectors:
        \[
            u = \cT v_T + \cN v_N,
        \]
        where $\cT, \cN \in \bbR^+$ represent the user's relative preference strengths for trend and niche content, respectively. Then we define $\bfu = (u)$.

        \item \textbf{Observation Noise $\bfeps$:} Since there is only one user, we slightly abuse notation and define $\bfeps = \{\epsilon_j\}_{j \in [\nT+\nN]}$, where each $\epsilon_j$ denotes independent observation noise. We assume the noise terms $\{\epsilon_j\}_{j \in [\nT+\nN]}$ are i.i.d. and bounded by $\barE$, satisfying 
        \[
            \barE < \frac{\nT - \nN}{\nT + \nN} \min(\cT, \cN).
        \]
        A canonical example is uniform noise: $\epsilon_j \overset{\text{i.i.d.}}{\sim} \mathcal{U}(-\barE, \barE)$.
    \end{enumerate}
\end{definition}




The first component of the $\TvN$ game, the initial content features, captures a key characteristic of real-world platforms where certain content attracts widespread attention while other content remains relatively niche~\citep{gill2007youtube, salganik2006experimental}. The second and third components, which describe user preferences and observation noise, are introduced to ensure analytical tractability. These assumptions are later relaxed in our synthetic experiments in \cref{sec:expe-syn}.

\subsection{Non-Strategic Platform Analysis} \label{sec:theories-non-strategic}
We begin our analysis of the $\TvN$ game in a non-strategic platform defined in \cref{def:static-platform}, where the optimal $\lambda$ is denoted as $\lambSta$. The outcome depends entirely on user preference, leading to two distinct conclusions presented in \cref{thrm:static-trend-lambda} for the trend-preferring case and \cref{thrm:static-niche-lambda} for the niche-preferring case.

\begin{theorem} \label{thrm:static-trend-lambda}
    In a non-strategic platform, consider a $\TvN$ game instance where the user prefers the trend group, i.e., $\cT > \cN$. Then there exists a \textbf{lower bound} $\lambStaL \ge 0$ on $\lambSta$, given by
    \[
    \lambStaL = 
        \begin{cases} 
            \frac{\nT\nN(\cN - \cT + 2\barE)}{\nT(\cT - \barE) - \nN(\cN + \barE)}, & \text{if } \barE > \frac{\cT - \cN}{2}, \\
            0, & \text{otherwise}.
        \end{cases}
    \]
    such that the expected user welfare in \cref{eq:lambsta-opti} is strictly increasing w.r.t. $\lambda$ for $\lambda < \lambStaL$ and remains constant for $\lambda \geq \lambStaL$. Consequently, any $\lambda \geq \lambStaL$ is optimal.

\end{theorem}
\begin{theorem} \label{thrm:static-niche-lambda}
    In a non-strategic platform, consider a $\TvN$ instance where the user prefers the niche group, i.e., $\cT < \cN$. Then there exists an \textbf{upper bound} $\lambStaU \ge 0$ on $\lambSta$, given by
    \[
        \lambStaU = 
        \begin{cases} 
            0, & \text{if } \barE \ge \frac{\cN - \cT}{2}, \\
            \frac{\nT\nN(\cN - \cT + 2\barE)}{\nT(\cT - \barE) - \nN(\cN + \barE)}, & \text{if } \frac{\nN\cN - \nT\cT}{\nT+\nN} < \barE < \frac{\cN - \cT}{2}, \\
            +\infty, & \text{otherwise}.
        \end{cases}
    \]
    such that the expected user welfare in \cref{eq:lambsta-opti} is strictly decreasing  w.r.t. $\lambda$ for $\lambda<\lambStaU$ and remains constant for $\lambda \geq \lambStaU$. Consequently, any $\lambda \in [0, \lambStaU]$ is optimal.
\end{theorem}
\begin{remark}
    Note that both the lower bound $\lambStaL$ in \cref{thrm:static-trend-lambda} and the upper bound $\lambStaU$ in \cref{thrm:static-niche-lambda} depend on the noise scale $\barE$. This dependence is expected: when $\barE$ is very small, the estimated feature $\hu$ tends to align closely with the ground-truth $u$ for any choice of $\lambda$, as indicated by \cref{eq:R-equals-US-plus-noise,eq:mf-hat-u}. In such cases, $\lambda$ has little impact on the overall welfare. Therefore, these bounds are informative only when the noise level $\barE$ exceeds a certain threshold.

    Taken together, \cref{thrm:static-trend-lambda} and \cref{thrm:static-niche-lambda} reveal a clear dichotomy for the optimal choice of $\lambda$ in a non-strategic platform, i.e., $\lambSta$. It tends to favor larger values as long as the user prefers the trend group and, conversely, smaller values when the user prefers the niche group.
\end{remark}

\subsection{Strategic Platform Analysis}
We then analyze the $\TvN$ setting under the strategic platform defined in \cref{def:competition-platform}, where the optimal regularization parameter is denoted by $\lambCom$. Our first step is to characterize the resulting pure Nash equilibrium, which serves as the foundation for analyzing the behavior of $\lambCom$.



\begin{proposition} \label{thrm:pne-analysis}
    Consider any $\TvN$ game with a given rewarding mechanism.
    Such a game is guaranteed to have a PNE. Furthermore, in any such equilibrium, the predicted user embedding coincides with the strategy of at least K content creator(s). Formally, we have $|\{ \hu \in \spne \}| \geq K$ where $\hu$ is the estimated user feature given by \cref{eq:mf-hat-u}.
\end{proposition}

This proposition suggests that, at the PNE, content creators tend to shift their strategies toward the estimated user feature $\hu$ in order to better align with user preferences. Building on this insight, the following theorem characterizes the platform’s optimal regularization parameter $\lambCom$.

\begin{theorem} \label{thrm:competition-lambda-upper-bound}
    In a strategic platform, for a $\TvN$ game instance, there always exists a finite upper bound $\lambComU$ on $\lambCom$ such that the expected user welfare in \cref{eq:lambcom-opti} is strictly decreasing w.r.t. $\lambda$ for all $\lambda > \lambComU$. Furthermore, if $\nT$ and $\nN$ both grow while maintaining a fixed ratio $\nT / \nN = C$ for some constant $C > 1$, then $\lambComU = O(\nT^{0.5 + \alpha})$ for any $\alpha > 0$.
\end{theorem}
\begin{remark}
    \cref{thrm:competition-lambda-upper-bound} indicates that the optimal choice of $\lambda$ in a strategic platform, i.e., $\lambCom$ inclines towards smaller values for any $\TvN$ game instance, regardless of user preference.
\end{remark}

We are now ready to compare the results derived from the strategic and non-strategic analyses, to elucidate the relationship between the optimal regularization strength across these two platform paradigms. For clarity, the comparison of theoretical results is still introduced by considering the trend-preferring and niche-preferring user, respectively.

\paragraph{Trend-preferring user} \cref{thrm:static-trend-lambda} establishes that for any trend-preferring user, the user welfare is maximized for any $\lambda$ at or above a lower bound, $\lambStaL$. Furthermore, for non-trivial noise levels (i.e. $\barE > \frac{\cT - \cN}{2}$), this bound grows as either $\nT$ or $\nN$ scales up. Assuming the content groups maintain a fixed size ratio ($\nT/\nN=C$), $\lambStaL$ exhibits linear growth with respect to the size of the trend group, i.e.,
\[
    \lambStaL = \Theta(\nT).
\]
As for the strategic platform, \cref{thrm:competition-lambda-upper-bound} establishes that the upper bound on $\lambda$ in the $\TvN$ game instance grows sub-linearly, i.e.,
\[
    \lambComU = O(\nT^{0.5 + \alpha}) \quad \text{for any } \alpha > 0.
\]
These two terms indicate that the growth rate of $\lambComU$ is slower than that of $\lambStaL$, implying that $\lambComU$ is strictly smaller than $\lambStaL$ when $\nT$ and $\nN$ are sufficiently large. Since real-world platforms typically involve a large number of content creators, this condition is readily satisfied in practice. Consequently, combining the results from \cref{thrm:static-trend-lambda,thrm:competition-lambda-upper-bound}, we conclude that $\lambCom < \lambSta$. This implies that for trend-preferring users, the platform should adopt a weaker regularization strategy (i.e., favor reduced bias) in strategic environments.




\paragraph{Niche-preferring user}
In contrast to the trend-preferring case, when the user prefers niche content, the optimal regularization parameter is upper bounded in both the non-strategic setting (denoted by $\lambStaU$) and the strategic setting (denoted by $\lambComU$), as established in \cref{thrm:static-niche-lambda} and \cref{thrm:competition-lambda-upper-bound}. This result suggests that the platform should adopt a relatively small regularization strength in both environments.



\paragraph{Insight for multi-user settings} In platform-wide environments with multiple users, aggregate preferences often skew toward trend content~\citep{xu2023smaa, fleder2009blockbuster}, which attracts a large share of user attention compared to niche content. This imbalance gives rise to the well-documented ``super-star'' effect in digital markets~\citep{salganik2006experimental, rosen1981superstar}. Such a market structure directly impacts platform policy design, as maximizing total user welfare requires prioritizing the preferences of the dominant, trend-oriented user base. Our stylized analysis offers concrete guidance: for a trend-preferring user, the optimal policy involves choosing a \textbf{smaller} regularization parameter $\lambda$. Extrapolating this insight to the multi-user setting, we conclude that a strategic platform should similarly adopt a \textbf{lower} regularization strength when managing the bias-variance tradeoff. This conclusion is further supported by our experiments on both synthetic and real-world datasets in \cref{sec:expe}.

\subsection{Intuitive Analysis} \label{sec:intuitive-analysis}
In this subsection, we offer an intuitive explanation of the theoretical results for \cref{sec:theories}. To illustrate our findings, \cref{fig:theory-results-comparison} shows examples of estimated user embeddings in an environment with $\nT=9$ trend and $\nN=1$ niche content. The observation noise is set to be $\epsilon_{j} \overset{\text{i.i.d.}}{\sim}\calN(0, 0.2^2)$ for $j \in [\nT+\nN]$. In this visualization, we define a \textit{trend-preferring user} with a ground-truth embedding of $(0.8, 0.6)$ and a \textit{niche-preferring user} with $(0.6, 0.8)$. The figure plots the distribution of estimated embeddings (blue dots) around these ground-truth values (large `X' markers). The dashed red line ($y=x$) separates the trend-preferring zone (below the line) from the niche-preferring zone (above the line).

\begin{figure}
    \centering
    \includegraphics[width=0.5\linewidth]{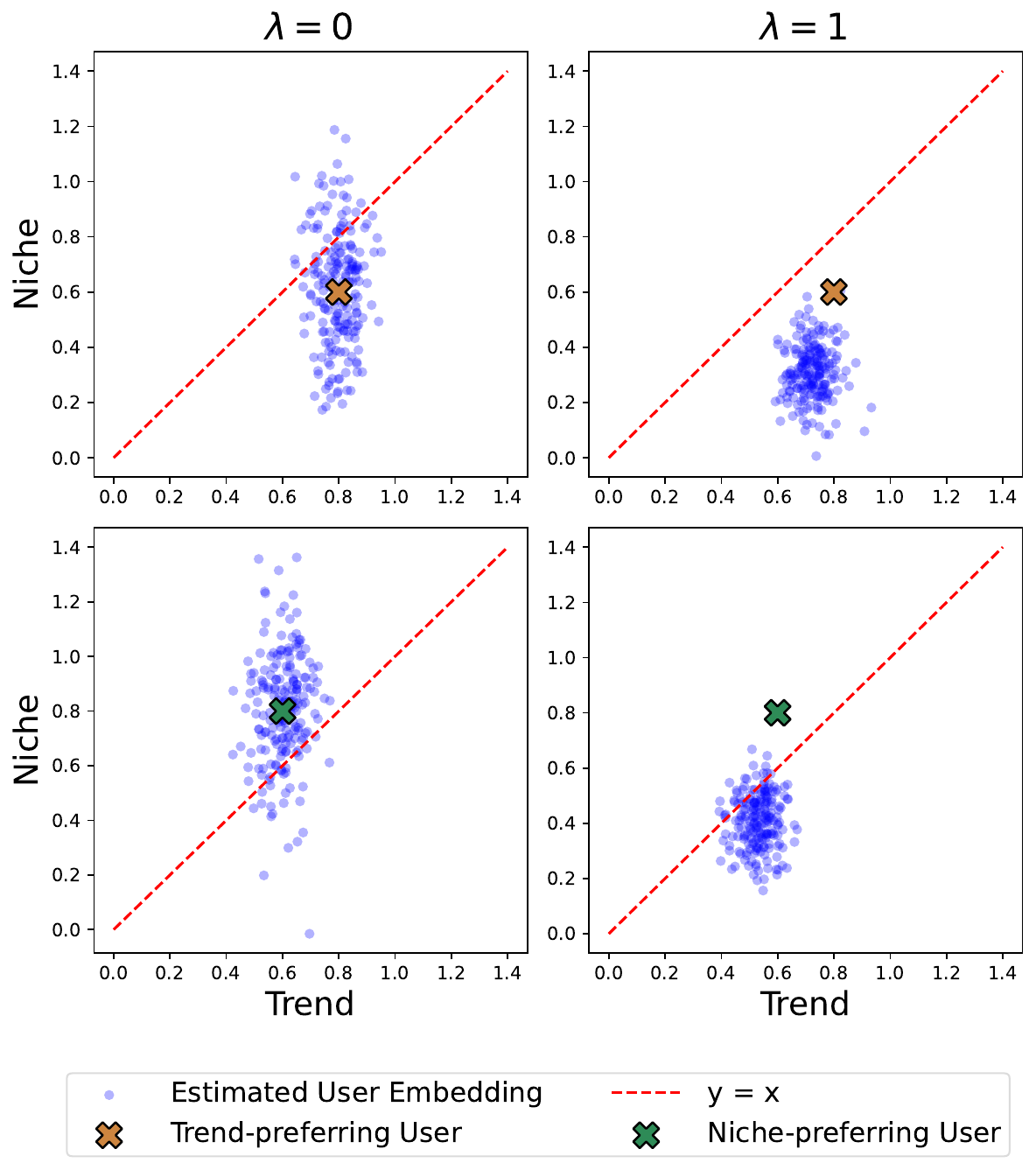}
    \vspace{-5px}
    \caption{Examples for distribution of estimated user embeddings from 200 trials under $\lambda=0$ and $\lambda=1$. The top row shows results for a trend-preferring user; the bottom row for a niche-preferring user. Ground-truth embeddings are marked by `X', and predictions by blue dots.
    Additional results with three more $\lambda$ values and an extra sample are provided in \cref{sec:omitted-figures}.}
    \vspace{-5px}
    \label{fig:theory-results-comparison}
\end{figure}

The left column shows, under weak regularization strength ($\lambda=0$), the estimations exhibit \textbf{high variance} with \textbf{low bias}. For both user types, a significant portion of the embeddings cross into the opposing preference zone, creating a risk of misidentifying the user’s primary preference.
However, in the right column, as the regularization strength gets stronger ($\lambda=1$), the estimations have \textbf{lower variance} but are \textbf{strongly biased} towards the trend-preference zone, pulling either the trend-preferring or niche-preferring user into the trend zone, which showcases the \textit{homogenizing effect}~\citep{anwar2024filter} of strong regularization.

These visualizations provide the intuition behind our main theorems.
For the non-strategic platform, welfare depends only on correctly classifying the user's preference. A trend-preferring user therefore \textit{welcomes} the homogenizing effect, as it secures their classification against noise. This explains why their welfare is maximized for any $\lambda$ above a lower bound ($\lambStaL$). Conversely, a niche-preferring user must \textit{avoid} this effect to prevent being misclassified, which explains the need for an upper bound on $\lambSta$.
For the strategic setting, welfare decreases with the expected estimation bias since, at equilibrium, the recommended content has the identical embedding to the user's biased estimate as implied by \cref{thrm:pne-analysis}. Here, the homogenizing effect is always harmful because it increases this expected bias, which explains why welfare decreases when $\lambda$ is too high, necessitating an upper bound for $\lambCom$.



\section{Experiments} \label{sec:expe}

We conduct a series of experiments to empirically validate our theoretical framework by simulating the $C^3_{BV}$ game on both synthetic and publicly available datasets.
To isolate the specific impact of strategic interactions, our core experimental design involves comparing the outcomes of the \textit{strategic environment} (illustrated by four distinct rewarding mechanisms, including $\BRCM$, $\Mexpo$, $\Menga$, and $\Mzero$~\citep{yao2023rethinking}) against a \textit{non-strategic baseline}. The following subsections detail the simulation setup and present the results of this comparison.

In practice, due to the analytical intractability of converging to a PNE in complex strategic settings, we follow \citep{yao2023howbad,yao2023rethinking,eilat2023performative,yu2025beyond} and approximate equilibrium outcomes via \emph{Local Better Response} (LBR) dynamics, with our welfare evaluation based on the resulting strategy profile. Although these dynamics typically converge to a local Nash equilibrium \citep{hron2023modeling,eilat2023performative}, we slightly abuse notation by using the globally-defined optimal parameter $\lambCom$ from \cref{eq:lambcom-opti}. The complete algorithmic procedure is detailed in \cref{sec:sec5-omitted}, accompanied by plots illustrating the user welfare evolution under these dynamics. The source code is publicly available at \url{https://github.com/wangkng/C3BV}.


\subsection{Experiments on Synthetic Data} \label{sec:expe-syn}

\begin{figure*}
    \centering
    \includegraphics[width=\linewidth]{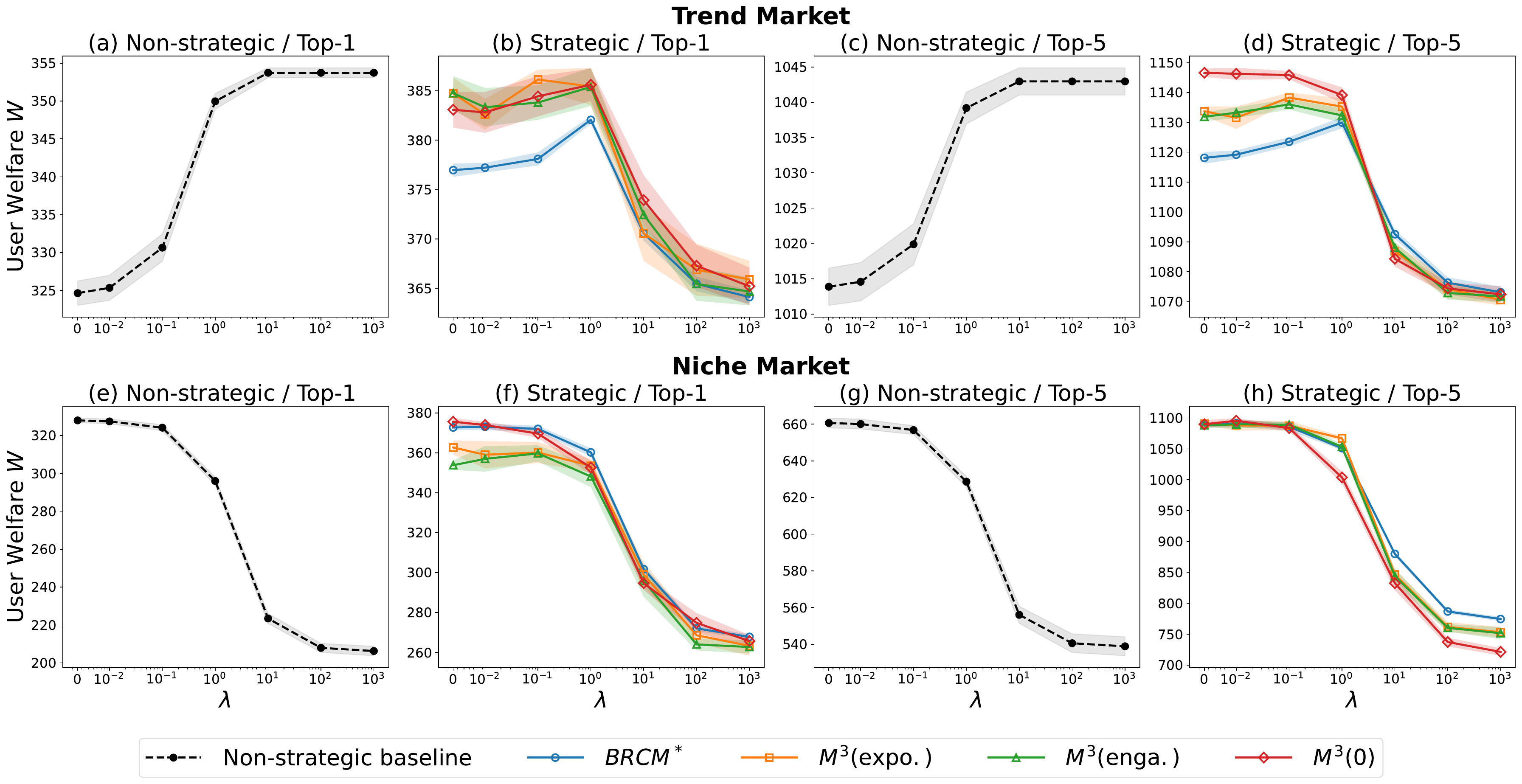}
    \caption{User welfare-$\lambda$ curve for Trend Market and Niche Market data in synthetic experiments. Error bars represent 0.5 standard deviation range, generated from simulations over 10 independent noise realizations.}
    \label{fig:synthetic-Welfare-all}
\end{figure*}

Our synthetic data experiments are designed to validate theoretical conclusions and test their robustness in a generalized environment that extends the stylized $\TvN$ setting defined in \cref{def:tvnEnvironment}.

\subsubsection{Synthetic Environments} We simulate a population of $n=10$ creators, whose strategies are $d=10$ dimensional vectors on the non-negative unit sphere (i. e., $\bbS^{d-1}_{\geq 0}$). The population is partitioned into a \textit{trend group} of $\nT=9$ creators and a \textit{niche group} of $\nN=1$, initially aligned along two randomly chosen orthonormal basis vectors, $v_T \in \bbS^{d-1}_{\geq 0}$ and $v_N \in \bbS^{d-1}_{\geq 0}$.

Based on these creator embeddings, we then construct a user population of $m=400$, with their \textit{ground-truth} features generated from the conic cone spanned by the trend (i.e., $v_T$) and niche (i.e., $v_N$) directions. To model distinct population-level preferences, we define two subspaces within this cone: a \textit{trend-preferring user space}, $\calC_T = \{ u = \alpha v_T + \beta v_N \mid \alpha > \beta \ge 0 \} \subset \bbRpos^d $, and a \textit{niche-preferring user space}, $\calC_N = \{ u = \alpha v_T + \beta v_N \mid \beta > \alpha \ge 0 \} \subset \bbRpos^d$. This setup provides an idealized construction of user preference.

Our synthetic experiments are built around two primary market scenarios. In the \textbf{Trend Market}, all $m$ users are randomly sampled from the subspace $\calC_T$, while in the \textbf{Niche Market}, users are drawn from the subspace $\calC_N$. Within each scenario, we conduct experiments for Top-$K$ recommendations with both $K=1$ and $K=5$.
We adopt a descending attention score sequence $\{r_k\}_{k=1}^K = \{\frac{1}{\log_2(2)}, \ldots, \frac{1}{\log_2(K+1)}\}$ to model the diminishing attention a user pays to lower-ranked items \citep{yao2023rethinking}. The observation noise is modeled as i.i.d. samples from $ \mathcal{N}(0, \sigma_e^2)$ with $\sigma_e = 0.5$.


\subsubsection{Results} We model the strategic environment as a dynamic process where each creator iteratively updates their strategy for a time horizon of $T=800$ steps, following the LBR algorithm \citep{yao2023rethinking,yao2024unveiling,yu2025beyond}. The result for the synthetic experiments is shown in \cref{fig:synthetic-Welfare-all}, including both Trend Market and Niche Market instances.

\begin{figure*}
    \centering
    \includegraphics[width=\linewidth]{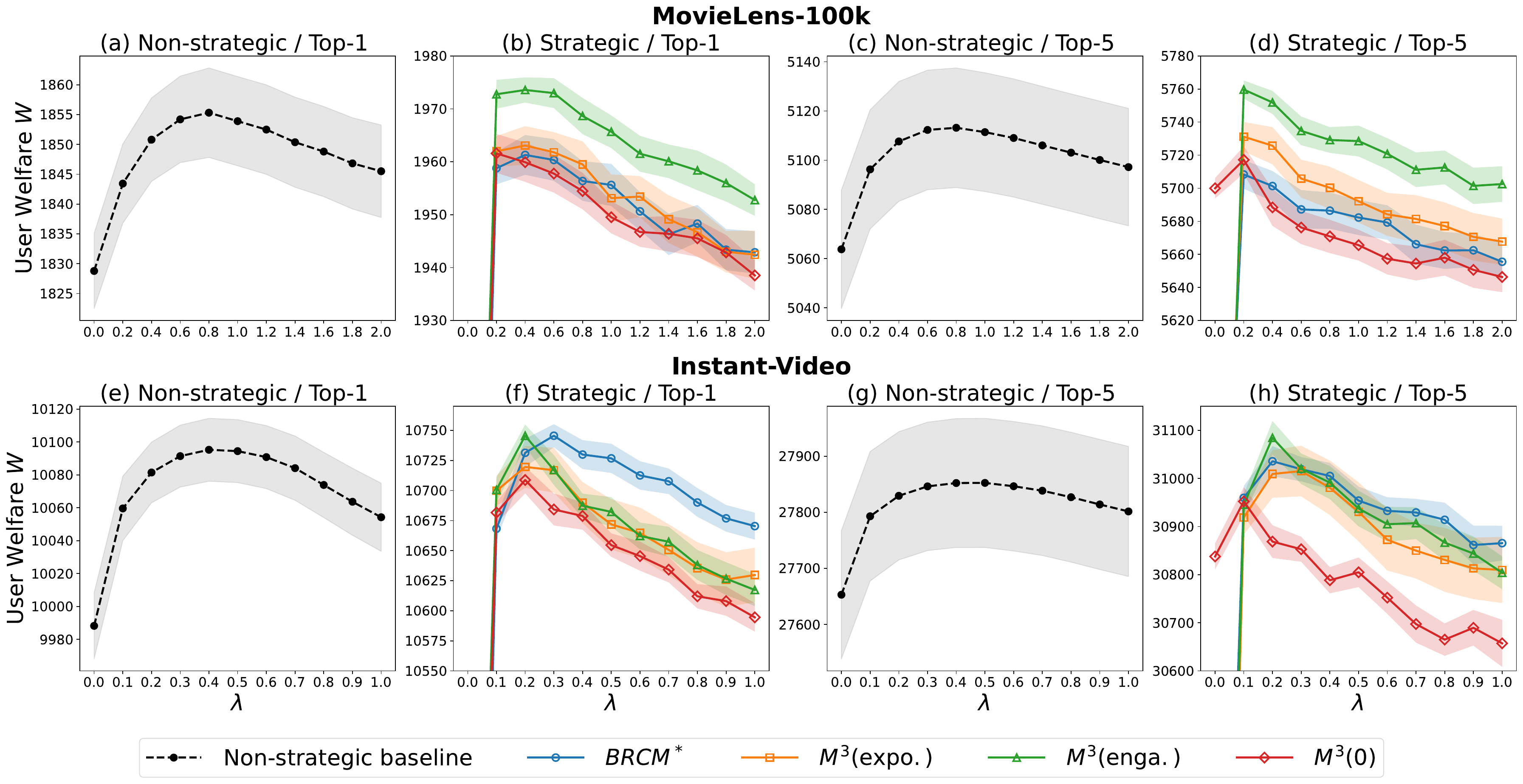}
    \caption{User welfare-$\lambda$ curve for experiments based on the MovieLens-100k and Instant-Video datasets. Error bars represent 0.2 standard deviation range, generated from simulations over 10 independent noise realizations. For better visibility, the $y$-axis in Plots (b), (d), (f), and (h) is truncated, as the value at $x = 0$ is significantly smaller and thus omitted.}
    \label{fig:real-Welfare-all}
\end{figure*}
 
\paragraph{Trend Market} \Cref{fig:synthetic-Welfare-all}~\hyperref[fig:synthetic-Welfare-all]{(a)--(d)} show how user welfare varies with~$\lambda$ in the Trend Market environment for top-$K$ recommendation ($K=1$ and $K=5$), including the non-strategic baseline and strategic results. The latter are illustrated at the last iteration step (i.e., Step $T$). These figures show that the optimal $\lambda$ of a strategic environment (i.e., $\lambCom$) is always \textit{smaller} than the optimal $\lambda$ in the non-strategic baseline (i.e., $\lambSta$).
Specifically, for $K=1$, the strategic optimum is $\lambCom=0.1$ for $\Mexpo$ and $\lambCom=1$ for three other mechanisms; and for $K=5$, the strategic optimum is $\lambCom=1$ for $\BRCM$, $\lambCom=0$ for $\Mzero$ and $\lambCom=0.1$ for the remaining mechanisms.
In contrast, the non-strategic baseline exhibits a clear preference for stronger regularization. For both $K=1$ and $K=5$, the optimal welfare is achieved for any $\lambda \ge 10$. This aligns with our theoretical finding in \cref{thrm:static-trend-lambda} that the optimal set is an unbounded interval $[\lambStaL, \infty)$, with $\lambStaL \leq 10$.

\paragraph{Niche Market} \Cref{fig:synthetic-Welfare-all}~\hyperref[fig:synthetic-Welfare-all]{(e)--(h)} present user welfare across different~$\lambda$ values in the Niche Market environment for top-$K$ recommendation ($K=1$ and $K=5$), including the non-strategic baseline and strategic results. These figures show that the optimal $\lambda$ of a strategic environment (i.e., $\lambCom$) shows no significant different from the optimal $\lambda$ in the non-strategic baseline (i.e., $\lambSta$), with both values being small and leaning toward 0.

These results are consistent with the theoretical insights presented in \cref{sec:theories} and extend the analysis to multi-user environments and settings with unbounded noise.


\subsection{Experiments on Real-World Datasets}

To assess the applicability of theoretical findings, we conduct simulations based on instances generated from two benchmark datasets: \textbf{MovieLens-100k}\footnote{\url{https://grouplens.org/datasets/movielens/100k}} and the Amazon Product Reviews dataset\footnote{\url{https://cseweb.ucsd.edu/~jmcauley/datasets/amazon/links}}. For the latter, we focus on the ``Instant Video'' category \citep{he2016ups}, which we refer to as \textbf{Instant-Video}.
The MovieLens-100k dataset comprises 943 users, 1,682 movies, and 100,000 ratings \citep{harper2015movielens}. For the Instant-Video dataset, we use the 5-core subset, which includes 5,130 users, 1,685 items, and 37,126 ratings.

\subsubsection{Environment Constructed from Real-world Datasets}
To establish the ground-truth basis for our experiments with real-world datasets, we derive a two-step process to obtain \textit{ground-truth} user and content embeddings. First, we derive a complete set of user and content embeddings from the full user-content interaction matrix (i.e., $R$), using Non-negative Matrix Factorization (NMF) \citep{lee1999nmf} with the embedding dimension $d=16$. Second, to align with the geometric assumptions of our theoretical model, we normalize each content embedding by its $\ell_2$-norm, effectively projecting it onto the unit hypersphere.
We use the complete (943 for MovieLens-100k and 5,130 for Instant-Video) user embeddings as the user population and randomly sample $n=10$ content as the strategic creators with their initial strategies identical to these embeddings, which constitutes the ground-truth embeddings for our simulations. For each dataset, we conduct experiments for Top-$K$ recommendations with both $K=1$ and $K=5$. The attention scores are the same as in \cref{sec:expe-syn}, and the observation noise is modeled as i.i.d. samples from $ \mathcal{N}(0, \sigma_e^2)$ with $\sigma_e = 0.3$.

\subsubsection{Results} The creators' strategy iteration also follows the LBR algorithm within a time horizon of $T=1500$ steps. The result for the experiments on real-world datasets is shown in \cref{fig:real-Welfare-all}, including both the MovieLens-100k and Instant-Video datasets.

The \Cref{fig:real-Welfare-all} shows how user welfare varies with~$\lambda$ for instances constructed from MovieLens-100k and Instant-Video, under top-$K$ recommendation, with $K=1$ and $K=5$. The figure includes both the non-strategic baseline and the strategic outcomes.
The latter are shown at the final iteration (Step~$T$). Across both datasets, these figures consistently demonstrate that the optimal $\lambda$ in the strategic environment is strictly \emph{smaller} than the optimal $\lambda$ in the non-strategic baseline, i.e., $\lambCom < \lambSta$.



\paragraph{MovieLens-100k} 
Results on the MovieLens-100k dataset are shown in plots~\hyperref[fig:real-Welfare-all]{(a)--(d)}. 
For $K=1$, the strategic optimum is $\lambCom=0.2$ for $\Mzero$ and $\lambCom=0.4$ for the other mechanisms, compared to the non-strategic optimum $\lambSta=0.8$.
For $K=5$, all mechanisms yield $\lambCom=0.2$, again compared to the non-strategic baseline with $\lambSta=0.8$.

\paragraph{Instant-Video} 
Results on the Instant-Video dataset are shown in plots~\hyperref[fig:real-Welfare-all]{(e)--(h)}. 
For $K=1$, the strategic optimum is $\lambCom=0.3$ for $\BRCM$ and $\lambCom=0.2$ for the other mechanisms, compared to the non-strategic optimum $\lambSta=0.4$.
For $K=5$, the strategic optimum is $\lambCom=0.3$ for $\Mexpo$, $\lambCom=0.1$ for $\Mzero$, and $\lambCom=0.2$ for the remaining mechanisms, again compared to the non-strategic baseline with $\lambSta=0.5$.

\section{Conclusion}

In this paper, we propose the $C^3_{BV}$ framework to examine how strategic content creation reshapes the platform’s optimal bias-variance tradeoff for user feature estimation. Our analysis shows that when creators behave strategically, the platform should adopt weaker regularization (i.e., lower bias), and that the divergence from the non-strategic baseline is driven by users’ inherent preferences. These findings are supported by both theoretical results and empirical validation on synthetic and real-world datasets, highlighting the need to treat regularization not just as a statistical choice, but as a policy lever influenced by creator competition.
\section*{Acknowledgements}
Renzhe Xu's research was supported by the National Key R\&D Program of China (No. 2023YFA1009500), the National Natural Science Foundation of China (No. 72442024), the Shanghai Sailing Program (No. 24YF2711600), and the CCF-Huawei Populus Grove Fund. Bo Li's research was supported by the National Natural Science Foundation of China (Nos. 72171131 and 72133002).

\bibliographystyle{plainnat}
\bibliography{references}

\clearpage
\appendix

\section{Omitted Figure} \label{sec:omitted-figures}
We further present a more comprehensive example---under the same setting---to support the intuitive analysis in \Cref{sec:intuitive-analysis}. The first two rows correspond to two distinct trend-preferring users, with ground-truth embeddings of $(0.8,\ 0.6)$ and $(0.96,\ 0.28)$, respectively. The latter two rows correspond to two distinct niche-preferring users, whose ground-truth embeddings are $(0.6,\ 0.8)$ and $(0.28,\ 0.96)$, respectively. The results in \cref{fig:theory-results-comparison-large} are consistent with our previous analysis, showing the mismatching tendency under weak regularization and homogenizing effect under strong regularization.

\begin{figure}[H]
    \centering
    \includegraphics[width=1\linewidth]{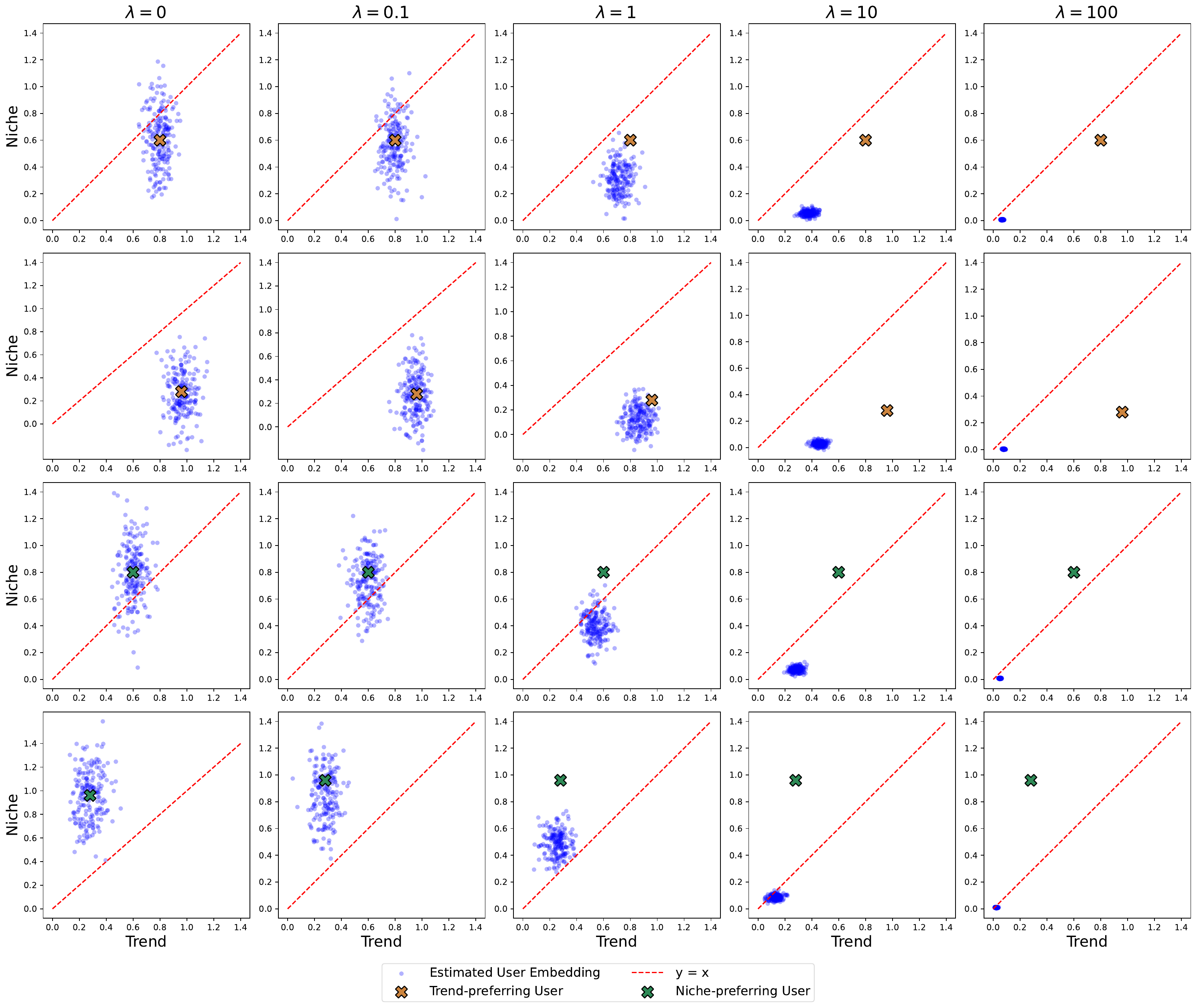}
    \caption{Examples for distribution of estimated user embeddings from 200 trials under five different $\lambda$s.}
    \label{fig:theory-results-comparison-large}
\end{figure}

\section{Additional Experimental Results} \label{sec:omitted-experiment}

To further examine the robustness of our findings, we report two sets of supplementary experiments.

\paragraph{(1) Varying data and modeling configurations.}
We first vary the experimental configuration along variant dimensions, including:
(i) the parameter values ($n,K$),
(ii) the dataset (Digital-Music\footnote{The ``Digital Music'' category of the Amazon Product Reviews corpus~\citep{he2016ups}. As with the Instant-Video dataset, we use the 5-score subset, consisting of 5,541 users, 3,568 items and 64,706 ratings.}), and
(iii) the feature-extraction baseline (DeepMF~\citep{xue2017dmf}).
The resulting optimal $\lambda$ values under the non-strategic baseline and each mechanism are summarized in \cref{tab:additional-configuration}.
\begin{table}[h]
\centering
\caption{The optimal $\lambda$ across mechanisms under additional experimental configurations.}
\label{tab:additional-configuration}
\begin{tabular}{lccccc}
\toprule
$\lambda^*$ & Non-Str & $\BRCM$ & $\Mexpo$ & $\Menga$ & $\Mzero$ \\
\midrule

\multicolumn{6}{l}{\textbf{Larger parameter values}} \\
$n=20,\,K=10$                  & 0.6 & 0.2 & 0.2 & 0.2 & 0.2 \\
$n=50,\,K=10$                  & 0.6 & 0.4 & 0.4 & 0.4 & 0.4 \\
$n=50,\,K=20$                  & 0.6 & 0.4 & 0.4 & 0.6 & 0.2 \\[2pt]

\multicolumn{6}{l}{\textbf{Digital-Music dataset}} \\
$n=10,\,K=1$     & 0.5 & 0.4 & 0.2 & 0.2 & 0.2 \\
$n=10,\,K=5$     & 0.5 & 0.2 & 0.2 & 0.2 & 0.1 \\[2pt]

\multicolumn{6}{l}{\textbf{DeepMF feature extraction}} \\
$n=10,\,K=1$           & 0.8 & 0.8 & 0.2 & 0.2 & 0.4 \\
$n=10,\,K=5$           & 0.8 & 0.6 & 0.4 & 0.2 & 0.2 \\
$n=20,\,K=10$          & 1.0 & 0.2 & 0.4 & 0.4 & 0.2 \\
\bottomrule
\end{tabular}
\end{table}

\paragraph{(2) Changing the platform objective.}
We additionally extend our analysis by replacing the platform objective (i.e., \emph{User Welfare}) with \emph{Nash Social Welfare} \citep{caragiannis2019unreasonableNSW}. Crucially, we re-run the full suite of experiments while strictly mirroring the configurations used in the main text.
The corresponding optimal $\lambda$ values are reported in \cref{tab:additional-NSW}.

\begin{table}[h]
\centering
\caption{The optimal $\lambda$ across mechanisms with another platform objective, the Nash Social Welfare.}
\label{tab:additional-NSW}
\begin{tabular}{lccccc}
\toprule
$\lambda^*$ & Non-Str & $\BRCM$ & $\Mexpo$ & $\Menga$ & $\Mzero$ \\
\midrule

\multicolumn{6}{l}{\textbf{Niche-Market}} \\
K=1 & 0.01 & 0.01 & 0.00 & 0.01 & 0.00 \\
K=5 & 0.00 & 0.01 & 0.00 & 0.00 & 0.01 \\[3pt]

\multicolumn{6}{l}{\textbf{Trend-Market}} \\
K=1 & $\ge 10$ & 1.0 & 0.0 & 1.0 & 1.0 \\
K=5 & $\ge 10$ & 1.0 & 0.1 & 0.1 & 0.0 \\[3pt]

\multicolumn{6}{l}{\textbf{MovieLens-100k}} \\
K=1 & 0.6 & 0.4 & 0.4 & 0.4 & 0.2 \\
K=5 & 0.6 & 0.2 & 0.2 & 0.2 & 0.2 \\[3pt]

\multicolumn{6}{l}{\textbf{Instant-Video}} \\
K=1 & 0.4 & 0.3 & 0.2 & 0.2 & 0.2 \\
K=5 & 0.5 & 0.2 & 0.3 & 0.2 & 0.1 \\

\bottomrule
\end{tabular}
\end{table}

\section{Omitted Proofs}

\paragraph{Additional notations} 
For the noise term $\epsilon$ introduced in \cref{def:tvnEnvironment}, we define $\barepsilonT$ as the average noise in the trend group, i.e., $\barepsilonT = \frac{\sum_{j=1}^{\nT} \epsilon_j}{\nT}$. Similarly, $\barepsilonN$ denotes the average noise in the niche group, i.e., $\barepsilonN = \frac{\sum_{j=\nT+1}^{\nT+\nN} \epsilon_j}{\nN}$. We define the function $F(\lambda; \barepsilonT, \barepsilonN) := \hu(\lambda)^\top(v_T - v_N)$.

\subsection{Proof of \cref{thrm:static-trend-lambda} and \cref{thrm:static-niche-lambda}}
We first introduce and prove the following lemmas, which will be instrumental in proving the main proposition.

\begin{lemma} \label{thrm:function-lambda-zero-point-derivative}
    For a fixed noise realization $(\bar{\epsilon}_T, \bar{\epsilon}_N)$, let $\lambda_0 > 0$ be any value such that $F(\lambda_0; \bar{\epsilon}_T, \bar{\epsilon}_N) = 0$. Then the derivative of $F$ w.r.t. $\lambda$ at $\lambda = \lambda_0$ is strictly positive, i.e.,
    \[
        \frac{\partial}{\partial \lambda} F(\lambda; \bar{\epsilon}_T, \bar{\epsilon}_N) \bigg|_{\lambda = \lambda_0} > 0.
    \]
\end{lemma}
\begin{proof}
    We start by expressing $F(\lambda; \bar{\epsilon}_T, \bar{\epsilon}_N)$ explicitly:
        $$F(\lambda; \bar{\epsilon}_T, \bar{\epsilon}_N)
        = \frac{\nT(\cT+\barepsilonT)}{\nT+\lambda} - \frac{\nN(\cN+\barepsilonN)}{\nN+\lambda}$$
    The condition $F(\lambda_0; \bar{\epsilon}_T, \bar{\epsilon}_N) = 0$ for some $\lambda_0 > 0$ implies the following equality:
    \begin{equation} \label{eq:function-equals-to-zero}
        \frac{\nT(\cT+\barepsilonT)}{\nT+\lambda_0} = \frac{\nN(\cN+\barepsilonN)}{\nN+\lambda_0}
    \end{equation}
    
    

    Now, we compute the derivative of $F$ with respect to $\lambda$, i.e.,
    \[
        \frac{\partial F}{\partial \lambda} = -\frac{\nT(\cT+\barepsilonT)}{(\nT+\lambda)^2} + \frac{\nN(\cN+\barepsilonN)}{(\nN+\lambda)^2}
    \]
    As a result,
    \begin{align*}
        \frac{\partial F}{\partial \lambda} \bigg|_{\lambda = \lambda_0}
        &= -\frac{1}{\nT+\lambda_0} \left( \frac{\nT(\cT+\barepsilonT)}{\nT+\lambda_0} \right) + \frac{\nN(\cN+\barepsilonN)}{(\nN+\lambda_0)^2} \\
        &= -\frac{1}{\nT+\lambda_0} \left( \frac{\nN(\cN+\barepsilonN)}{\nN+\lambda_0} \right) + \frac{\nN(\cN+\barepsilonN)}{(\nN+\lambda_0)^2} \tag{by \cref{eq:function-equals-to-zero}} \\
        &= \frac{\nN(\cN+\barepsilonN)}{\nN+\lambda_0} \left( \frac{1}{\nN+\lambda_0} - \frac{1}{\nT+\lambda_0} \right) \\
        &> 0. \tag{By the assumption that $\barE<\frac{\nT-\nN}{\nT+\nN}\cdot\min(\cT,\cN)$}
    \end{align*}
    Now the claim follows.
\end{proof}

\begin{lemma} \label{thrm:function-lambda-monotonicity}
    For each $\lambda \ge 0$, define the set
    \[
        E(\lambda) := \left\{ (\bar{\epsilon}_T, \bar{\epsilon}_N) \mid F(\lambda; \bar{\epsilon}_T, \bar{\epsilon}_N) > 0 \right\}.
    \]
    Then \( E(\lambda) \) is monotonically expanding in $\lambda$, i.e., for any $\lambda_2 > \lambda_1 \ge 0$, we have $E(\lambda_1) \subseteq E(\lambda_2)$.
\end{lemma}
\begin{proof}
    To prove the claim, we take an arbitrary element $(\bar{\epsilon}_T, \bar{\epsilon}_N) \in E(\lambda_1)$ and show it must also be in $E(\lambda_2)$ for any $\lambda_2 > \lambda_1$.
    By definition, $(\bar{\epsilon}_T, \bar{\epsilon}_N) \in E(\lambda_1)$ means $F(\lambda_1; \bar{\epsilon}_T, \bar{\epsilon}_N) > 0$. We want to show $F(\lambda_2; \bar{\epsilon}_T, \bar{\epsilon}_N) > 0$.

    Let's assume for the sake of contradiction that $F(\lambda_2; \bar{\epsilon}_T, \bar{\epsilon}_N) \le 0$.
    Consider $F(\lambda)$ as a function of $\lambda$ for this fixed noise realization. $F(\lambda)$ is continuous for $\lambda \ge 0$. We have $F(\lambda_1) > 0$ and $F(\lambda_2) \le 0$. By the Intermediate Value Theorem, there must exist some $\lambda_0 \in (\lambda_1, \lambda_2]$ such that $F(\lambda_0) = 0$.

    Let $\lambda_0^* = \inf \{ \lambda \in (\lambda_1, \lambda_2] \mid F(\lambda) = 0 \}$. Since $F$ is continuous and $F(\lambda_1) > 0$, we know that for all $\lambda \in [\lambda_1, \lambda_0^*)$, $F(\lambda) > 0$. This means the function approaches the root $\lambda_0^*$ from above (or is constant at zero), which implies that its derivative at that point must be non-positive:
    \[
        \frac{\partial}{\partial \lambda} F(\lambda; \bar{\epsilon}_T, \bar{\epsilon}_N) \bigg|_{\lambda = \lambda_0^*} \le 0.
    \]
    This directly contradicts the result of \cref{thrm:function-lambda-zero-point-derivative}, which states that the derivative at any such root must be strictly positive.

    The contradiction stems from our assumption that $F(\lambda_2; \bar{\epsilon}_T, \bar{\epsilon}_N) \le 0$. Therefore, this assumption must be false, and we must have $F(\lambda_2; \bar{\epsilon}_T, \bar{\epsilon}_N) > 0$.
    This shows that $(\bar{\epsilon}_T, \bar{\epsilon}_N) \in E(\lambda_2)$, and thus $E(\lambda_1) \subseteq E(\lambda_2)$.
\end{proof}

Now we are ready to prove the two propositions.

We begin by writing explicitly the user welfare function under static environment for top-$K$ recommendation, i.e., $W(\bfv; \calG(\lambda,\bfv,\bfeps))$, which is the expected user utility over all noise realizations:
\begin{equation} \label{eq:Wsta}
    \begin{aligned}
        W(\bfv; \calG(\lambda,\bfv,\bfeps))
        &= \bbE_{\epsilon_1,\ldots, \epsilon_{\nT+\nN}} \left[\sum_{k=1}^K r_k \cdot \left(
            v_T^\top u \cdot \bbI[\hu^\top v_T > \hu^\top v_N]
            + v_N^\top u \cdot \bbI[\hu^\top v_T \leq \hu^\top v_N]\right)
        \right] \\
        &= \sum_{k=1}^K r_k \left[P(\lambda) \cdot v_T^\top u + (1-P(\lambda)) \cdot v_N^\top u \right],
    \end{aligned}
\end{equation}
where $P(\lambda) \triangleq \Pr\left[\hu^\top v_T > \hu^\top v_N\right]$.

Then we proceed to prove \cref{thrm:static-trend-lambda}.
\begin{proof}[Proof of \cref{thrm:static-trend-lambda}]

     Since the user prefers the trend group ($\cT > \cN$), it holds that $v_T^\top u > v_N^\top u$. Thus, $W(\bfv; \calG(\lambda,\bfv,\bfeps))$ is a strictly increasing linear function of the probability $P(\lambda)$, and maximizing welfare is equivalent to maximizing $P(\lambda)$.

    The monotonicity of the user welfare follows from \cref{thrm:function-lambda-monotonicity}. The lemma states that the set $E(\lambda) = \{(\barepsilonT, \barepsilonN) : F(\lambda; \barepsilonT, \barepsilonN) > 0\}$ is monotonically expanding in $\lambda$. This directly implies that its probability measure, $P(\lambda)$, is a monotonically non-decreasing function of $\lambda$. Consequently, $W(\bfv; \calG(\lambda,\bfv,\bfeps))$ is also monotonically non-decreasing, thus implying $\lambda \to \infty$ always an optimal choice.
    

    
    An optimal solution set of the form $[\lambStaL, \infty)$ exists if and only if the non-decreasing welfare function $W(\bfv; \calG(\lambda,\bfv,\bfeps))$ becomes constant for all $\lambda \ge \lambStaL$. This requires the probability $P(\lambda)$ to saturate at its supremum, $P_{\max}$, at a finite $\lambStaL$. This saturation is guaranteed only if $P_{\max}=1$, which means the inequality $F(\lambda)>0$ must hold for all possible noise realizations.
    
    To ensure this, we test against the worst-case noise for this objective, which occurs when $(\barepsilonT, \barepsilonN) = (-\barE, \barE)$. We seek the smallest $\lambda \ge 0$ for which the inequality
    \[
        \frac{\nT(\cT-\barE)}{\nT+\lambda} > \frac{\nN(\cN+\barE)}{\nN+\lambda}
    \]
    holds. By rearrange the inequality, we obtain:
    \[
        \lambda\left[\nT(\cT-\barE) - \nN(\cN+\barE)\right] > \nT\nN\left[(\cN+\barE) - (\cT-\barE)\right].
    \]
    Let $K_T = \nT(\cT - \barE) - \nN(\cN + \barE)$ denote the coefficient of $\lambda$, and let $C_T = \nT\nN(\cN - \cT + 2\barE)$ be the constant on the right-hand side. The solution depends on the signs of $K_T$ and $C_T$, which leads to the following case analysis.
    
    \begin{enumerate}
        \item \textbf{$K_T \le 0$}. In this scenario, $P(\lambda)$ never reaches 1, and the welfare function is strictly increasing towards a supremum that is never attained at any finite $\lambda$. The optimal solution is at infinity, so we define $\lambStaL = +\infty$. However, this condition never satisfies since $\barE < \frac{\nT - \nN}{\nT + \nN} \min(\cT, \cN) $.
    
        \item \textbf{$K_T > 0$ and $C_T > 0$}. This condition is equivalent to $\frac{\nN}{\nT} < \frac{\cT - \barE}{\cN + \barE} < 1$. The inequality holds for all $\lambda$ above a certain threshold. The optimal solution set is $[\lambStaL, \infty)$, where
        \[
            \lambStaL = \frac{\nT\nN(\cN - \cT + 2\barE)}{\nT(\cT - \barE) - \nN(\cN + \barE)}.
        \]
        
        \item \textbf{$C_T \le 0$}. This condition is equivalent to $\frac{\cT - \barE}{\cN + \barE} \ge 1$. The inequality holds true for all $\lambda \ge 0$, which means the trend group always dominates. The welfare is always maximized, making the optimal set $[0, \infty)$. The lower bound of this set is $\lambStaL = 0$.
    \end{enumerate}
    
    To unify all cases, we formally define the lower bound $\lambStaL$ as follow:
    \[
        \lambStaL = 
        \begin{cases} 
            \frac{\nT\nN(\cN - \cT + 2\barE)}{\nT(\cT - \barE) - \nN(\cN + \barE)}, & \text{if } \frac{\nN}{\nT} < \frac{\cT - \barE}{\cN + \barE} < 1 \\
            0, & \text{if } \frac{\cT - \barE}{\cN + \barE} \ge 1.
        \end{cases}
    \]
    As a result, a lower bound $\lambStaL$ always exists, such that any $\lambda \ge \lambStaL$ represents an optimal choice. The claim then follows.
\end{proof}

Now we prove \cref{thrm:static-niche-lambda}.
\begin{proof}[Proof of \cref{thrm:static-niche-lambda}]
    The proposition's premise is that the user prefers the niche group, i.e., $\cT < \cN$. Thus, to maximize the welfare function $W(\bfv; \calG(\lambda,\bfv,\bfeps))$ defined in \cref{eq:Wsta}, is equivalent to minimizing the probability $P(\lambda)$, which indicates the probability that a trend content is recommended to this user.

    As established in the proof of \cref{thrm:static-trend-lambda}, the probability function $P(\lambda) = \Pr[F(\lambda)>0]$ is monotonically non-decreasing for all $\lambda \ge 0$. Since the user prefers the niche group ($\cT < \cN$), welfare is maximized by minimizing $P(\lambda)$. The minimum value $P(0)$ is achieved at $\lambda=0$, making $\lambSta=0$ an optimal choice. We now analyze the upper bound $\lambStaU$ such that the welfare $W(\lambda;s)$ remains at its maximum for all $\lambda \in [0,\lambStaU]$.

    An interval $[0, \lambStaU]$ is optimal if and only if $P(\lambda)$ is constant on this interval. The strongest form of this is when the probability is at its absolute minimum, $P(\lambda)=0$. We therefore seek the largest interval $[0, \lambStaU]$ for which we can guarantee $P(\lambda)=0$.
    
    To guarantee $P(\lambda)=0$, the condition $F(\lambda; \barepsilonT, \barepsilonN) \le 0$ must hold for \textit{all} possible noise realizations $\varepsilon \in [-\barE, \barE]$. To ensure this, we must test against the worst-case noise realization, which occurs when $(\barepsilonT, \barepsilonN) = (\barE, -\barE)$.
    
    We seek the range of $\lambda \ge 0$ for which the inequality
    \[
        \frac{\nT(\cT+\barE)}{\nT+\lambda} \le \frac{\nN(\cN-\barE)}{\nN+\lambda}
    \]
    holds. By rearrange the inequality, we obtain:
    \[
        \lambda\left[\nT(\cT+\barE) - \nN(\cN-\barE)\right] \le \nT\nN\left[(\cN-\barE) - (\cT+\barE)\right].
    \]
    Let $K_N = \nT(\cT + \barE) - \nN(\cN - \barE)$ denote the coefficient of $\lambda$, and let $C_N = \nT\nN[(\cN-\barE) - (\cT+\barE)]$ be the constant on the right-hand side. The feasibility of the inequality depends on the signs of $K_N$ and $C_N$, which leads to the following case analysis.

    \begin{enumerate}
        \item \textbf{$K_N \le 0$}. This condition is equivalent to $\frac{\cN - \barE}{\cT + \barE} \ge \frac{\nT}{\nN} >1$, implying the constant $C_N$ is positive. In this case, the inequality is satisfied for all $\lambda \ge 0$, meaning that under any noise realization, the trend group never dominates. Consequently, $P(\lambda) = 0$ for all $\lambda$, the user welfare is always maximized, and the optimal solution set is the entire non-negative real line: $[0, \infty)$. We define $\lambStaU = +\infty$. 
    
        \item \textbf{$K_N > 0$ and $C_N > 0$}. This condition is equivalent to $1 < \frac{\cN - \barE}{\cT + \barE} < \frac{\nT}{\nN}$. The inequality holds for all $\lambda$ below a threshold. In this case, maximal welfare is guaranteed for all $\lambda \in [0, \lambStaU]$, where
        \[
            \lambStaU = \frac{\nT\nN[(\cN-\barE) - (\cT+\barE)]}{\nT(\cT + \barE) - \nN(\cN - \barE)}.
        \]
    
        \item \textbf{$C_N \le 0$}. This condition is equivalent to $\frac{\cN - \barE}{\cT + \barE} \le 1$. In this case, the inequality can only be satisfied by $\lambda = 0$. Therefore, for any $\lambda > 0$, we can no longer guarantee that $P(\lambda) = 0$ holds for all noise realizations. As a result, $P(\lambda)$ becomes strictly increasing for $\lambda > 0$, and the optimal solution is unique at $\lambda = 0$. We thus define $\lambStaU = 0$.
    \end{enumerate}
    To unify all cases above, we formally define the upper bound $\lambStaU$ as follow:
    \[
        \lambStaU = 
        \begin{cases} 
            +\infty, & \text{if } \frac{\cN - \barE}{\cT + \barE} \geq \frac{\nT}{\nN} \\
            \frac{\nT\nN(\cN - \cT + 2\barE)}{\nT(\cT - \barE) - \nN(\cN + \barE)}, & \text{if } 1 < \frac{\cN - \barE}{\cT + \barE} < \frac{\nT}{\nN} \\
            0, & \text{if } \frac{\cN - \barE}{\cT + \barE} \le 1 .
        \end{cases}
    \]
    As a result, an upper bound $\lambStaU$ always exists such that any $\lambda \in [0, \lambStaU]$ guarantees maximal user welfare. The claim then follows directly.
\end{proof}

\subsection{Proof of \cref{thrm:pne-analysis}}
\begin{proof}
    We first establish the existence of a pure strategy Nash equilibrium (PNE).
    Consider the symmetric strategy profile $\bfs^0 = (\hu, \ldots, \hu)$, where all creators adopt strategy $\hu$.
    Suppose creator $j$ deviates unilaterally to some alternative strategy $s'_j \neq \hu$, while all other creators stick to $\hu$. It suffices to show that:
    $$\sigma(s'_j,\hu) < \sigma(\hu,\hu).$$
    Given that the rewarding mechanism $RM$ is individually monotone, it follows that:
    \[
        RM(\sigma'_j(\hu), \sigma_{-j}(\hu)) \leq RM(\sigma_j(\hu), \sigma_{-j}(\hu)).
    \]
    According to the definition for creator utility defined in \cref{eq:creator_j-utility}, we have $\pi_j(s'_j,s_{-j}) \leq \pi_t(s_j,s_{-j})$. Therefore, the profile $\bfs^0$ constitutes a PNE.

    We prove the claim $|\{ \hu \in \spne \}| \geq K $ by contradiction. Assume there exists a PNE profile, $\spne = (s_1, \dots, s_m)$, where $|\{ \hu \in \spne \}| = k' < K $.
    Since $k' < K \leq \nT + \nN $, there must be at least one creator, indexed by $t$, whose strategy $s_t$ in the equilibrium profile is not $\hu$ (i.e., $s_t \neq \hu$).

    Consider a unilateral deviation by creator $t$ to a new strategy $s'_t = \hu$, while all other creators' strategies $s_{-t}$ remain fixed. By definition, the estimated user embedding $\hu$ is the vector that maximizes the matching score function $\sigma(\cdot, \hu)$. Because $s_t \neq \hu$, the matching score for creator $t$ strictly increases with this deviation:
        $$\sigma(s'_t, \hu) > \sigma(s_t, \hu).$$
    Given that the rewarding mechanism $RM$ is individually monotone, it follows that:
    \[
        RM(\sigma_t(\hu), \sigma_{-t}(\hu)) \geq RM(\sigma_t'(\hu), \sigma_{-t}(\hu)).
    \]
    
    According to the definition for creator utility defined in \cref{eq:creator_j-utility}, we have $\pi_t(s_t,s_{-t}) \geq \pi_t(s^\prime_t,s_{-t})$.
    This shows that creator $t$ has a profitable unilateral deviation. However, this contradicts our initial premise that $\spne$ is a PNE.
    Therefore, our initial assumption that $k' < K$ must be false. This completes the proof.
\end{proof}

\subsection{Proof of \cref{thrm:competition-lambda-upper-bound}}

By \cref{thrm:pne-analysis}, in any PNE, the predicted user embedding coincides with the strategy of at least one content creator. As a result, the top-$K$ content recommended to the user always has the same embedding as $\hu$. Hence, the user welfare at the PNE, denoted as $W(\spne; \calG)$ here,  is given by the cosine similarity between the estimated and true user preference vectors:
\begin{equation}
    \begin{aligned}
        W(\spne; \calG)
        &= \bbE_{\bfeps}\Bigl[
            W\bigl(\spne(\calG(\lambda, \bfv, \bfeps)); \calG(\lambda, \bfv, \bfeps)\bigr)
          \Bigr] \\
        &= \bbE_{\bfeps}\Biggl[
            \sum_{k=1}^K r_k \cdot \frac{\hu^\top u}{\|\hu\|_2}
          \Biggr] \\
        &= \bbE_{\bfeps}\Biggl[
            \frac{
              \sum_{k=1}^K r_k \Bigl(
                  \cT \frac{\nT(\cT+\barepsilonT)}{\nT+\lambda}
                + \cN \frac{\nN(\cN+\barepsilonN)}{\nN+\lambda}
              \Bigr)
            }{
              \left(
                \left(\frac{\nT(\cT+\barepsilonT)}{\nT+\lambda}\right)^2
                +
                \left(\frac{\nN(\cN+\barepsilonN)}{\nN+\lambda}\right)^2
              \right)^{1/2}
            }
          \Biggr].
    \end{aligned}
\end{equation}
Denote the following random variable
\[
A(\lambda) = \frac{\frac{\nT(\cT+\barepsilonT)}{\nT+\lambda}}{\frac{\nN(\cN+\barepsilonN)}{\nN+\lambda}} = \frac{\nT(\nN + \lambda)(\cT + \barepsilonT)}{\nN(\nT+\lambda)(\cN+\barepsilonN)}.
\]
Then the derivative $\frac{\partial}{\partial \lambda} W(\spne(\calG(\lambda,\bfv,\bfeps)); \calG(\lambda,\bfv,\bfeps))$ with respect to $\lambda$, denoted as $g(\lambda, s)$, is
\begin{equation} \label{eq:gradient-welfare-str}
    \begin{aligned}
        \, g(\lambda, s) 
        = & \, \frac{\partial}{\partial \lambda}W(\spne; \calG) \\
        = & \, \bbE\left[\frac{
            \sum_{k=1}^K r_k\nT\nN(\nT-\nN)
            (\cT+\barepsilonT)(\cN+\barepsilonN)
            \Biggl[
            \begin{aligned}
                &\nN\cT(\cN+\barepsilonN)(\nT+\lambda) \\
                &\qquad\qquad\qquad\qquad {} - \nT\cN(\cT+\barepsilonT)(\nN+\lambda)
            \end{aligned}
            \Biggr]
        }{
            (\nT+\lambda)^3(\nN+\lambda)^3
            \left[
                \left(\frac{\nT(\cT+\barepsilonT)}{\nT+\lambda}\right)^2
                +
                \left(\frac{\nN(\cN+\barepsilonN)}{\nN+\lambda}\right)^2
            \right]^{3/2}
        }\right] \\
        = & \, \frac{\sum_{k=1}^K r_k(\nT-\nN)\cN}{(\nT+\lambda)(\nN+\lambda)} \cdot \bbE\left[\frac{
            \frac{\nT(\nN+\lambda)
            (\cT+\barepsilonT)
            }{\nN(\nT+\lambda)(\cN+\barepsilonN)}\left[\frac{\cT}{\cN} - \frac{\nT(\cT+\barepsilonT)(\nN+\lambda)}{\nN(\nT+\lambda)(\cN+\barepsilonN)}\right]
        }{
            \left[
                \left(\frac{\nT(\nN+\lambda)(\cT+\barepsilonT)}{\nN(\nT+\lambda)(\cN+\barepsilonN)}\right)^2 + 1
            \right]^{3/2}
        }\right] \\
        = & \, \frac{\sum_{k=1}^K r_k(\nT-\nN)\cN}{(\nT+\lambda)(\nN+\lambda)} \cdot \bbE\left[\frac{A(\lambda)\left(\frac{\cT}{\cN} - A(\lambda)\right)}{\left(A^2(\lambda) + 1\right)^{3/2}}\right].
    \end{aligned}
\end{equation}

We first show that there exists a finite upper bound $\lambComU$ on $\lambCom$ such that the user welfare function $W(\spne;\calG)$ is strictly decreasing w.r.t $\lambda$ for all $\lambda > \lambComU$.
\begin{lemma} \label{lemma:tvn-strategic-existence}
    In a $\TvN$ game instance, there exists a finite upper bound $\lambComU$ on $\lambCom$ such that the user welfare function $W(\spne;\calG)$ is strictly decreasing w.r.t $\lambda$ for all $\lambda > \lambComU$.
\end{lemma}
\begin{proof}
    Note that $\barepsilonN$ and $\barepsilonT$ is bounded in $[-\barE, \barE]$ and $\barE < \frac{\nT-\nN}{\nT+\nN}\min(\cT, \cN)$. Therefore, there must exist a constant $\alpha > 0$ such that $\barE \le \left(\frac{\nT-\nN}{\nT+\nN} - \alpha\right)\min(\cT, \cN)$. As a result, it holds that
    \[
    \begin{aligned}
    A(\lambda)
    &\ge 
    \frac{\nT(\nN + \lambda)\left(\cT - \left(\frac{\nT-\nN}{\nT+\nN} - \alpha\right)\min(\cT, \cN)\right)}
         {\nN(\nT + \lambda)\left(\cN + \left(\frac{\nT-\nN}{\nT+\nN}\min(\cT, \cN) - \alpha\right)\right)}  \\
    &\ge 
    \frac{\nT(\nN + \lambda)\cT\left(1 + \alpha - \frac{\nT-\nN}{\nT+\nN}\right)}
         {\nN(\nT + \lambda)\cN\left(1 - \alpha + \frac{\nT-\nN}{\nT+\nN}\right)}
    = 
    \frac{\nT(\nN+\lambda)\cT(2\nN + \alpha(\nT + \nN))}
         {\nN(\nT+\lambda)\cN(2\nT-\alpha(\nT+\nN))}.
    \end{aligned}
    \]
    Since the right-hand side of the above equation is increasing and
    \[
    \lim_{\lambda \rightarrow \infty} \frac{\nT(\nN+\lambda)\cT(2\nN + \alpha(\nT + \nN))}{\nN(\nT+\lambda)\cN(2\nT-\alpha(\nT+\nN))} 
    = \frac{\nT\cT(2\nN + \alpha(\nT + \nN))}{\nN\cN(2\nT-\alpha(\nT+\nN))} 
    > \frac{\nT\cT \cdot 2\nN}{\nN\cN \cdot 2\nT} = \frac{\cT}{\cN},
    \]
    there must exist a constant $\lambComU$ such that $A(\lambda) > \cT / \cN$ for any values of $\barepsilonT$ and $\barepsilonN$. As a result, \cref{eq:gradient-welfare-str} indicates that the gradient $g(\lambda, s) < 0$, and hence the user welfare function $W(\spne;\calG)$ is strictly decreasing with respect to $\lambda$ for all $\lambda > \lambComU$.
\end{proof}

Then we show that $\lambComU$ is $O(\nT^{0.5 + \alpha})$ for any $\alpha > 0$.

\begin{lemma} \label{lemma:tvn-strategic-order}
    Assuming that $\nT / \nN = C$ for some constant $C > 1$, we have that $\lambComU = O(\nT^{0.5 + \alpha})$ for any $\alpha > 0$.
\end{lemma}
\begin{proof}
    Fix any $0 < \alpha < 0.5$. Denote the good event $G$ as
    \begin{equation} \label{eq:good-event}
        G = \bbI\left[\barepsilonT \ge - \barE \cdot \nT^{\alpha - 0.5} \quad \text{and} \quad \barepsilonN \le \barE \cdot \nN^{\alpha - 0.5}\right].
    \end{equation}
    Let $\bar{G}$ be the complement of $G$. By the union bound and Hoeffding's inequality, it holds that
    \[
    \begin{aligned}
        \bbP(\bar{G}) & = \bbP\left(\barepsilonT < - \barE \cdot \nT^{\alpha - 0.5} \quad \text{or} \quad \barepsilonN > \barE \cdot \nN^{\alpha - 0.5}\right) \\
        & \le \bbP\left(\barepsilonT < - \barE \cdot \nT^{\alpha - 0.5}\right) + \bbP\left(\barepsilonN > \barE \cdot \nN^{\alpha - 0.5}\right) \\
        & \le \exp\left(-\frac{\nT^{2\alpha}}{2}\right) + \exp\left(-\frac{\nN^{2\alpha}}{2}\right)
    \end{aligned}
    \]
    Define $B(\lambda)$ as the term inside the expectation in \cref{eq:gradient-welfare-str}, i.e.,
    \[
    B(\lambda) = \frac{A(\lambda)\left(\frac{\cT}{\cN} - A(\lambda)\right)}{\left(A^2(\lambda) + 1\right)^{3/2}}.
    \]
    When $\bar{G}$ happens, it holds that
    \[
    B(\lambda) \le \frac{\cT^2}{4\cN^2}.
    \]
    Let $\lambComU$ be the solution $\lambda$ to the following equation:
    \[
    \frac{\nT(\nN + \lambda)\left(\cT - \bar{E} \cdot \nT^{\alpha - 0.5}\right)}{\nN(\nT + \lambda)\left(\cN + \bar{E} \cdot \nN^{ \alpha - 0.5}\right)} = \left(1 + \frac{1}{\nT}\right)\frac{\cT}{\cN}.
    \]
    We have that
    \begin{equation} \label{eq:lambcomu}
    \lambComU = \frac{\nT\nN\left(\barE(1 + 1 / \nT)\cT\nN^{\alpha-0.5} + \barE\cN\nT^{\alpha - 0.5} + \cN\cT/\nT\right)}{\cT\cN(\nT - \nN(1 + 1/\nT)) - \barE\cN\nT^{\alpha + 0.5} - \barE(1+1/\nT)\cT\nN^{\alpha+0.5}}.
    \end{equation}
    It holds that $\lambComU > 0$ when $\nT$ is large enough and $\lambComU$ is $O\left(\nT^{0.5+\alpha}\right)$. In addition, we have that
    \[
    A(\lambda) = \frac{\nT(\nN + \lambda)(\cT + \barepsilonT)}{\nN(\nT+\lambda)(\cN+\barepsilonN)} \le C \cdot \frac{\cT + \barE}{\cN - \barE}.
    \]
    As a result, when the good event $G$ in \cref{eq:good-event} happens and $\lambda \ge \lambComU$, we have that for any realization of the random variables, it holds that
    \[
    A(\lambda) \ge \frac{\nT(\nN + \lambda)\left(\cT - \bar{E} \cdot \nT^{\alpha - 0.5}\right)}{\nN(\nT + \lambda)\left(\cN + \bar{E} \cdot \nN^{ \alpha - 0.5}\right)} = \left(1 + \frac{1}{\nT}\right)\frac{\cT}{\cN}.
    \]
    and
    \[
    B(\lambda) \le - \frac{\frac{\cT^2}{\cN^2}\cdot \frac{1}{\nT}}{\left(\left(C \cdot \frac{\cT + \barE}{\cN - \barE}\right)^2+1\right)^{3/2}} = -\frac{\cT^2}{\cN^2}\frac{C_1}{\nT}
    \]
    where $C_1 = \left(\left(C \cdot \frac{\cT + \barE}{\cN - \barE}\right)^2+1\right)^{-3/2}$ is a constant.

    Therefore, it holds that when $\nT$ is large enough and $\lambda \ge \lambComU$, it holds that
    \[
    \begin{aligned}
        \, \bbE[B(\lambda)] 
        = & \, \bbE[B(\lambda) \mid G]\bbP(G) + \bbE[B(\lambda) \mid \bar{G}]\bbP(\bar{G}) \\
        \le & \, \left(-\frac{\cT^2}{\cN^2}\frac{C_1}{\nT}\right)\left(1 - \exp\left(-\frac{\nT^{2\alpha}}{2}\right) - \exp\left(-\frac{\nN^{2\alpha}}{2}\right)\right) + \frac{\cT^2}{4\cN^2} \cdot \left(\exp\left(-\frac{\nT^{2\alpha}}{2}\right) + \exp\left(-\frac{\nN^{2\alpha}}{2}\right)\right) \\
        = & \, -\frac{\cT^2}{\cN^2}\frac{C_1}{\nT} + \left(\frac{\cT^2}{\cN^2}\frac{C_1}{\nT} + \frac{\cT^2}{4\cN^2}\right)\cdot \left(\exp\left(-\frac{\nT^{2\alpha}}{2}\right) + \exp\left(-\frac{\nN^{2\alpha}}{2}\right)\right).
    \end{aligned}
    \]
    It is clear that the right-hand side of the above equation becomes negative when $\nT$ is sufficiently large. This implies that
    \[
    g(\lambda, s) = \frac{\sum_{k=1}^K r_k(\nT - \nN)\cN}{(\nT + \lambda)(\nN + \lambda)} \cdot \bbE[B(\lambda)] < 0,
    \]
    and hence $W(\spne(\calG(\lambda,\bfv,\bfeps)); \calG(\lambda,\bfv,\bfeps))$ is strictly decreasing for all $\lambda \ge \lambComU$ when $\nT$ is large enough. Note that the threshold $\lambComU$ in \cref{eq:lambcomu} satisfies $\lambComU = O(\nT^{\alpha + 0.5})$. The claim then follows.
\end{proof}

Then \cref{thrm:competition-lambda-upper-bound} follows from \cref{lemma:tvn-strategic-existence,lemma:tvn-strategic-order}.

\section{Supplementary Material for \cref{sec:expe}} \label{sec:sec5-omitted}

\subsection{Local Better Response Algorithm} 
The pseudo-code of the Local Better Response~ \citep{yao2024user,yu2025beyond} Algorithm is given in \cref{alg:lbr}.
\begin{algorithm}[H]
    \caption{(\textbf{LBR}) Local Better Response update at time $t$}
    \label{alg:lbr}
    \begin{algorithmic}[1]
        \State \textbf{Input:} Learning rate $\eta$, current joint strategy profile $\boldsymbol{s}^{(t)}=(s_1^{(t)},\dots,s_n^{(t)})$.
        \State Sample a random direction $\boldsymbol{g}_j \in \mathbb{S}^{d-1}$ for player $j$.
        \State \textbf{Propose and project:}
        $\tilde{s}_j \leftarrow \Pi_{\mathbb{S}^{d-1}}\!\left(s_j^{(t)} + \eta\,\boldsymbol{g}_j\right)$.
        \If{$\pi_j\!\left(\tilde{s}_j,\, \boldsymbol{s}^{(t)}_{-j}\right) \ge \pi_j(\boldsymbol{s}^{(t)})$}
            \State $s_j^{(t+1)} \leftarrow \tilde{s}_j$.
        \Else
            \State $s_j^{(t+1)} \leftarrow s_j^{(t)}$.
        \EndIf
    \end{algorithmic}
\end{algorithm}

\subsection{User Welfare Evolution under LBR Dynamics} \label{sec:dynamics-omitted}

\begin{figure}[H]
    \centering
    \begin{subfigure}{0.44\linewidth}
        \centering
        \includegraphics[width=\linewidth]{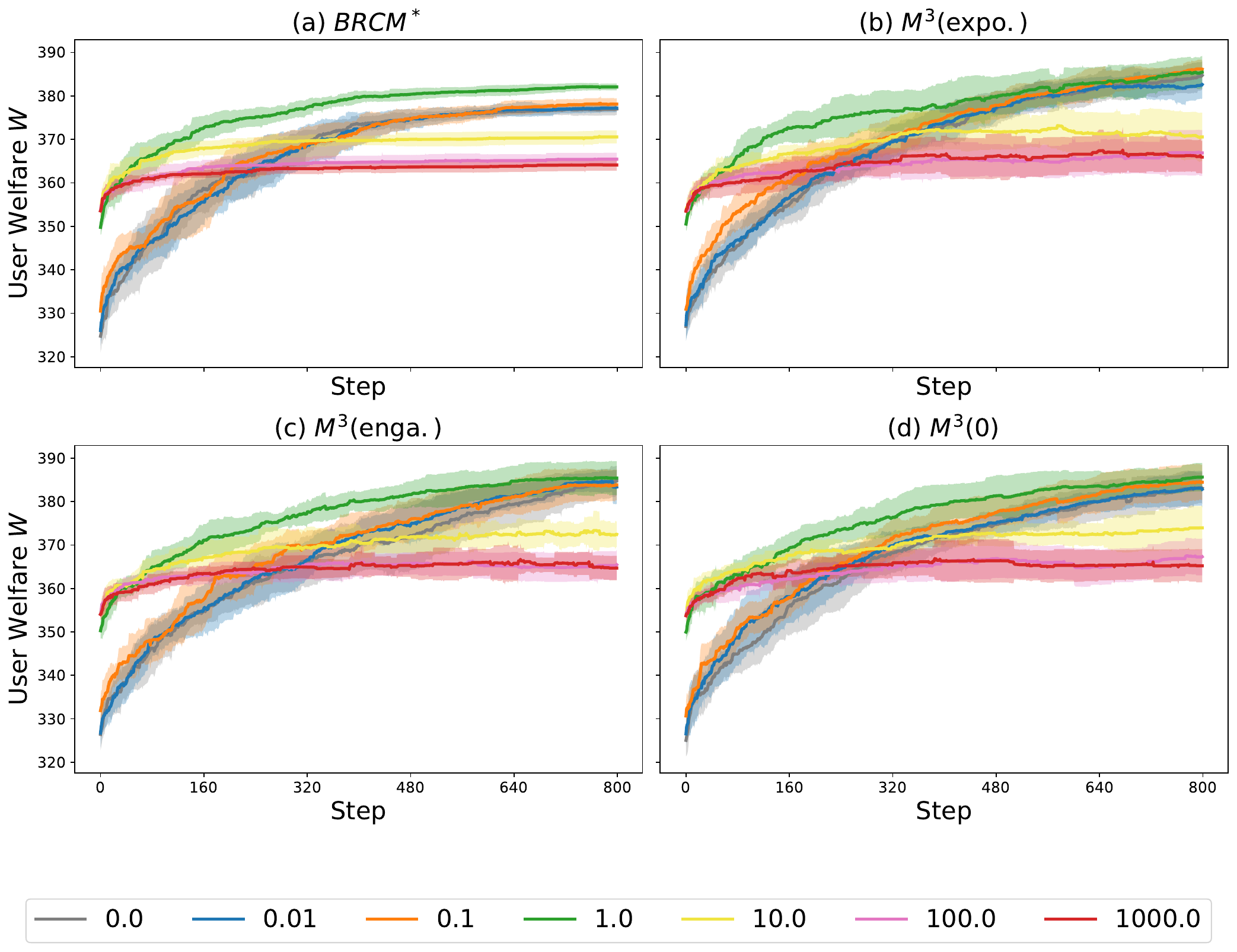}
        \caption{$K=1$.}
        \label{fig:uni-tvn9-pure-top1}
    \end{subfigure}
    \hfill 
    \begin{subfigure}{0.44\linewidth}
        \centering
        \includegraphics[width=\linewidth]{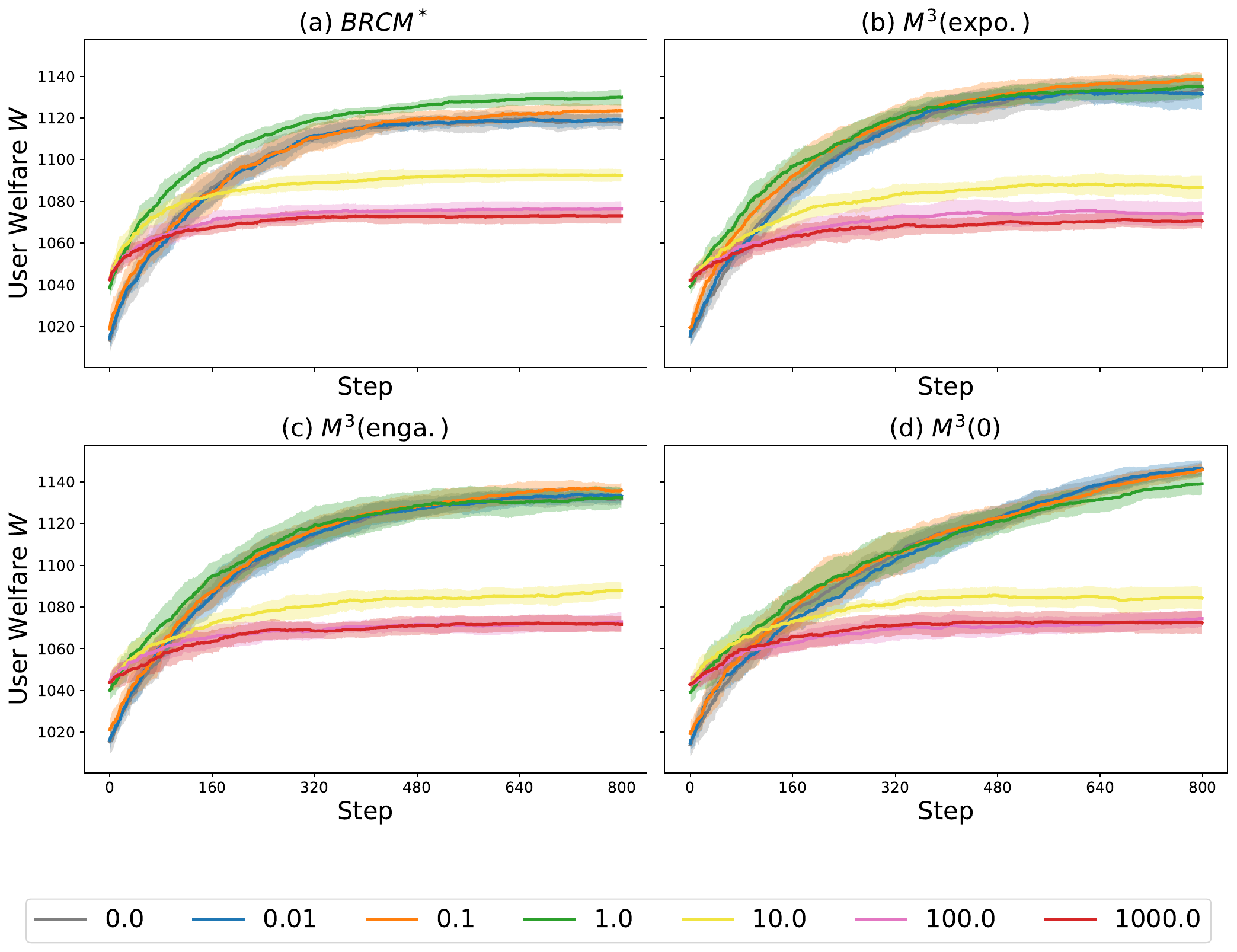}
        \caption{$K=5$.}
        \label{fig:uni-tvn9-pure-top5}
    \end{subfigure}

    \caption{User welfare evolution on the synthetic Trend Market dataset with $K=1$ and $K=5$.}
    \label{fig:uni-tvn9-pure-combined}
\end{figure}

\begin{figure}[H]
    \centering
    \begin{subfigure}{0.44\linewidth}
        \centering
        \includegraphics[width=\linewidth]{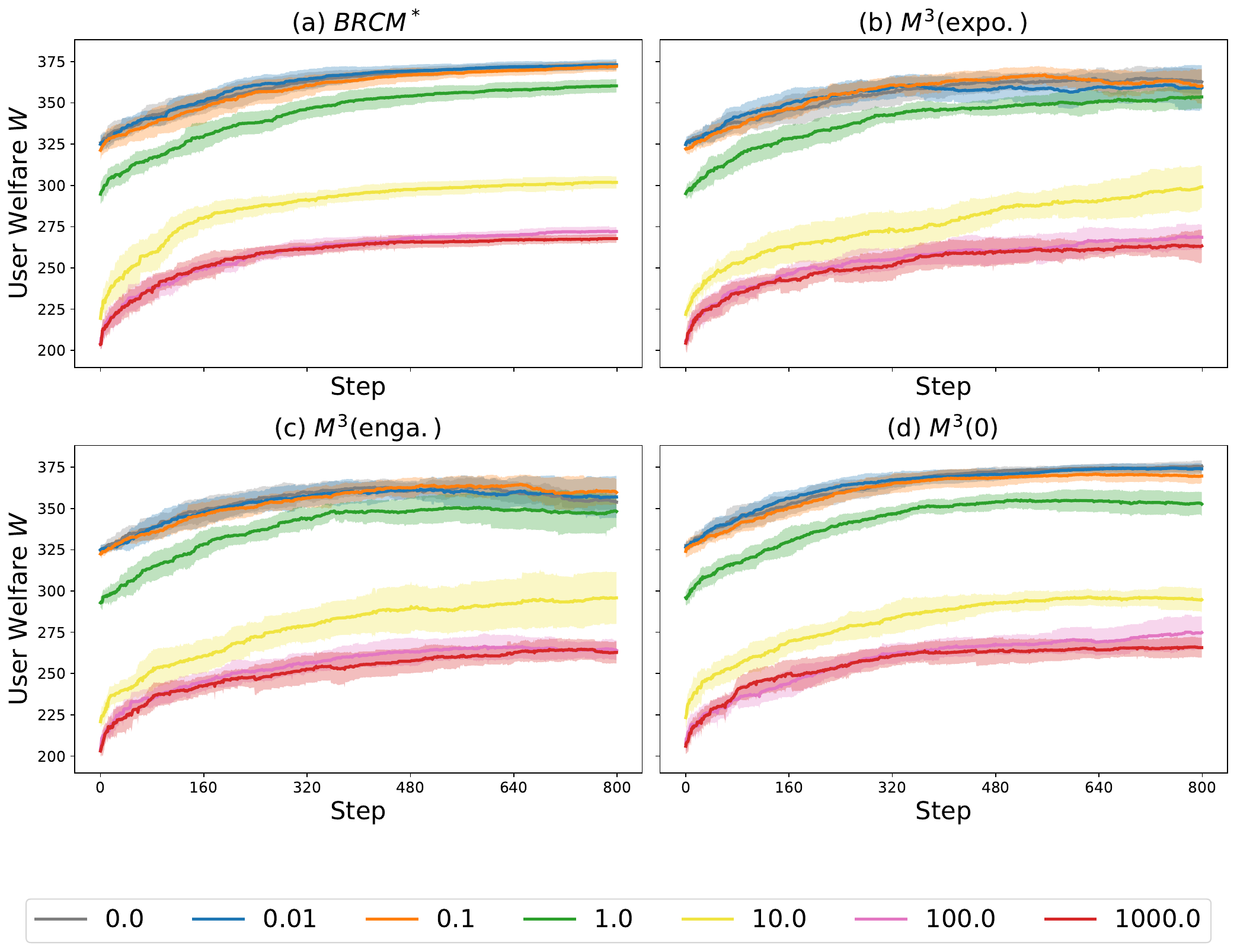}
        \caption{$K=1$.}
        \label{fig:uni-tvn9-reverse-pure-top1}
    \end{subfigure}
    \hfill 
    \begin{subfigure}{0.44\linewidth}
        \centering
        \includegraphics[width=\linewidth]{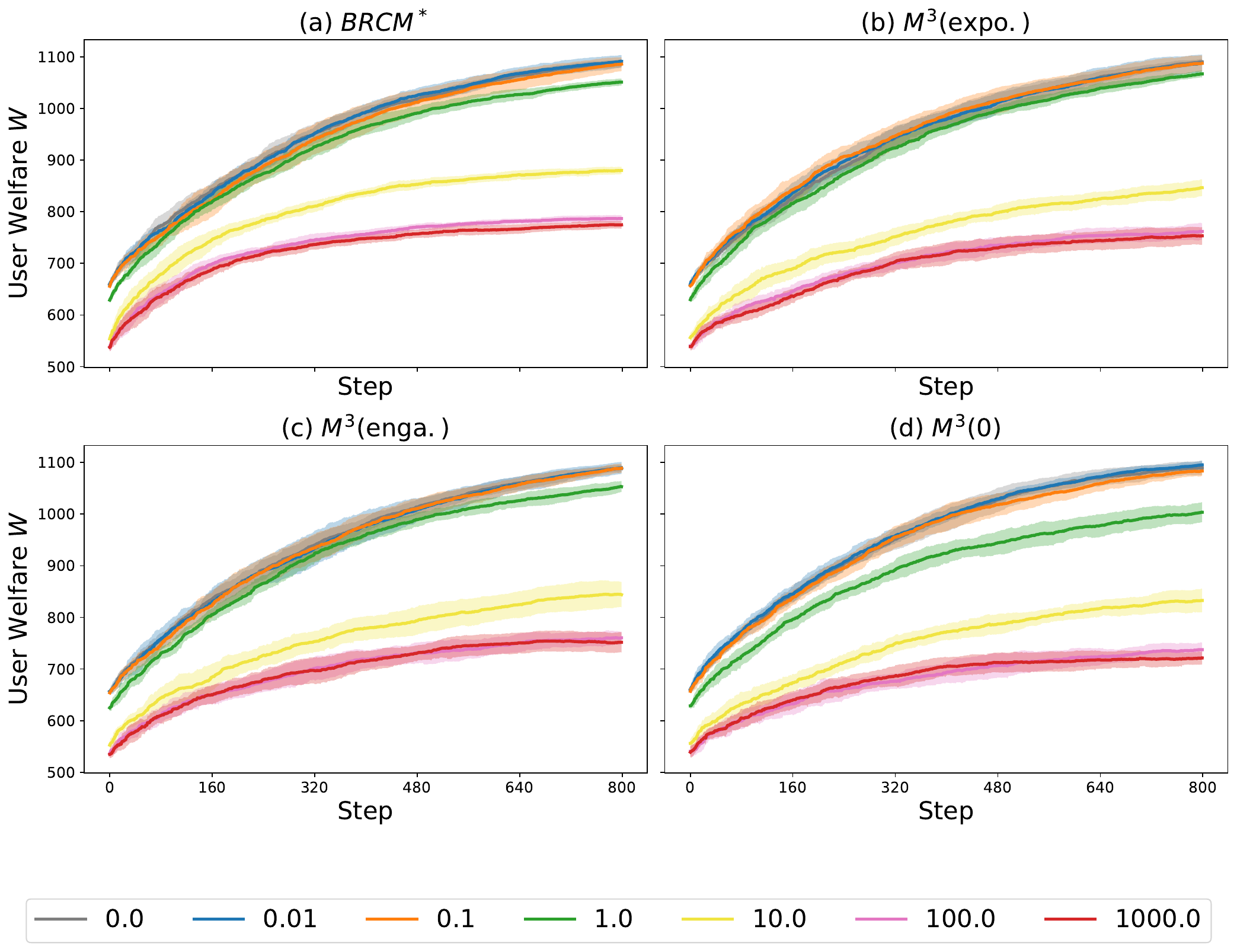}
        \caption{$K=5$.}
        \label{fig:uni-tvn9-reverse-pure-top5}
    \end{subfigure}
    \caption{User welfare evolution on the synthetic Niche Market dataset with $K=1$ and $K=5$.}
    \label{fig:uni-tvn9-reverse-pure-combined}
\end{figure}

\begin{figure}[H]
    \centering
    \begin{subfigure}{0.44\linewidth}
        \centering
        \includegraphics[width=\linewidth]{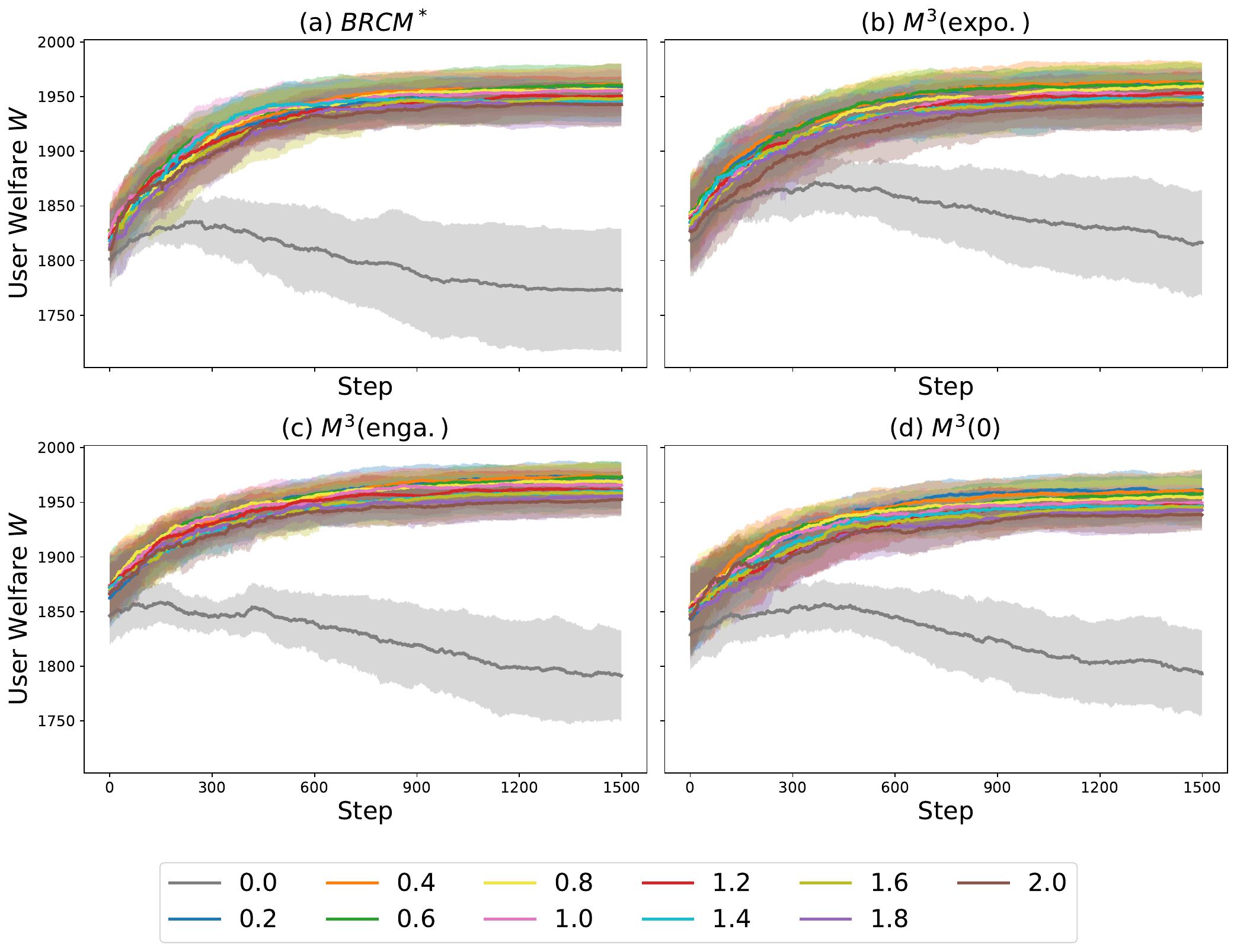}
        \caption{$K=1$.}
        \label{fig:ml-100k-top1}
    \end{subfigure}
    \hfill
    \begin{subfigure}{0.44\linewidth}
        \centering
        \includegraphics[width=\linewidth]{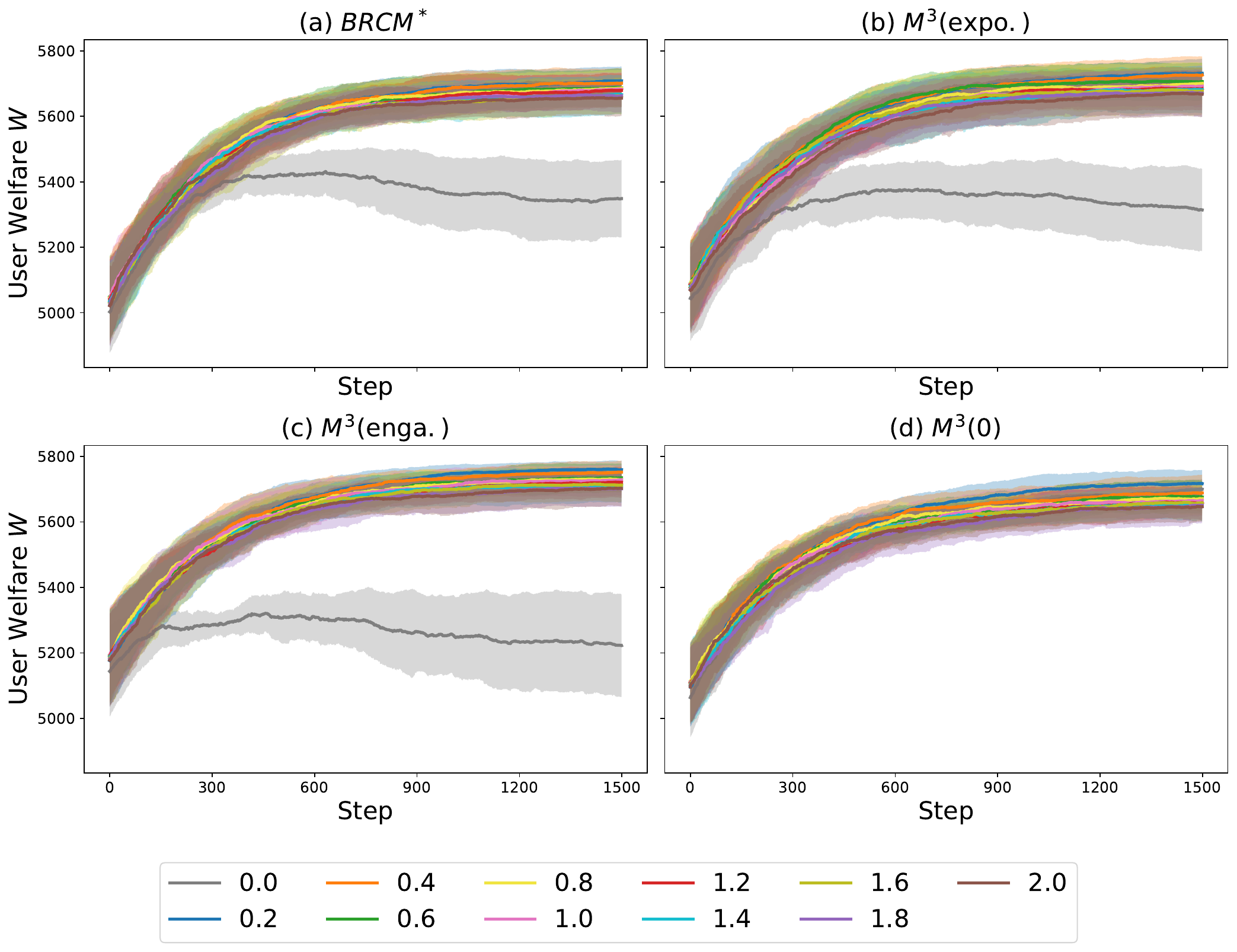}
        \caption{$K=5$.}
        \label{fig:ml-100k-top5}
    \end{subfigure}
    \caption{User welfare evolution on the MovieLens-100k dataset with $K=1$ and $K=5$.}
    \label{fig:ml-100k-combined}
\end{figure}

\begin{figure}[H]
    \centering
    \begin{subfigure}{0.44\linewidth}
        \centering
        \includegraphics[width=\linewidth]{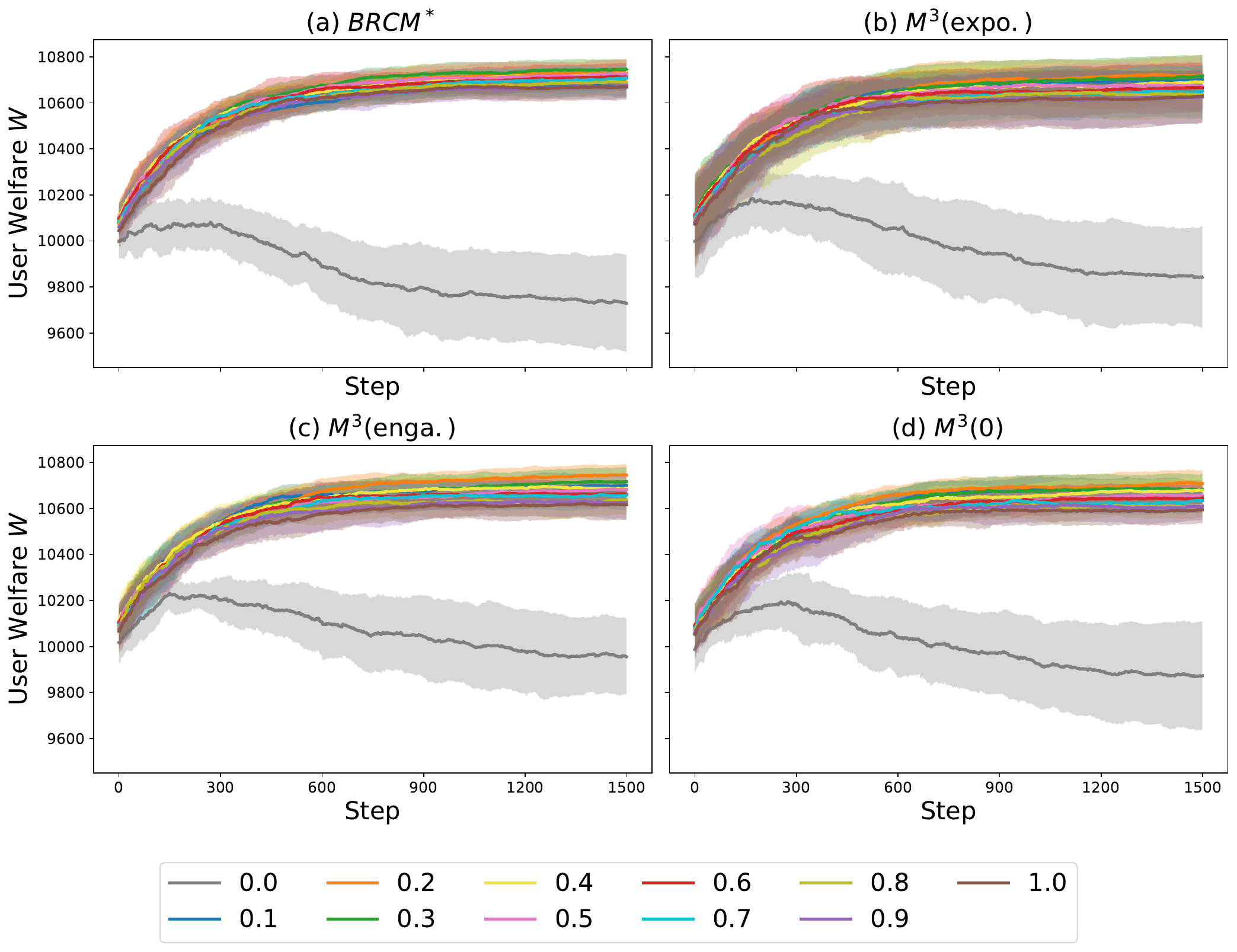}
        \caption{$K=1$.}
        \label{fig:amz-top1}
    \end{subfigure}
    \hfill 
    \begin{subfigure}{0.44\linewidth}
        \centering
        \includegraphics[width=\linewidth]{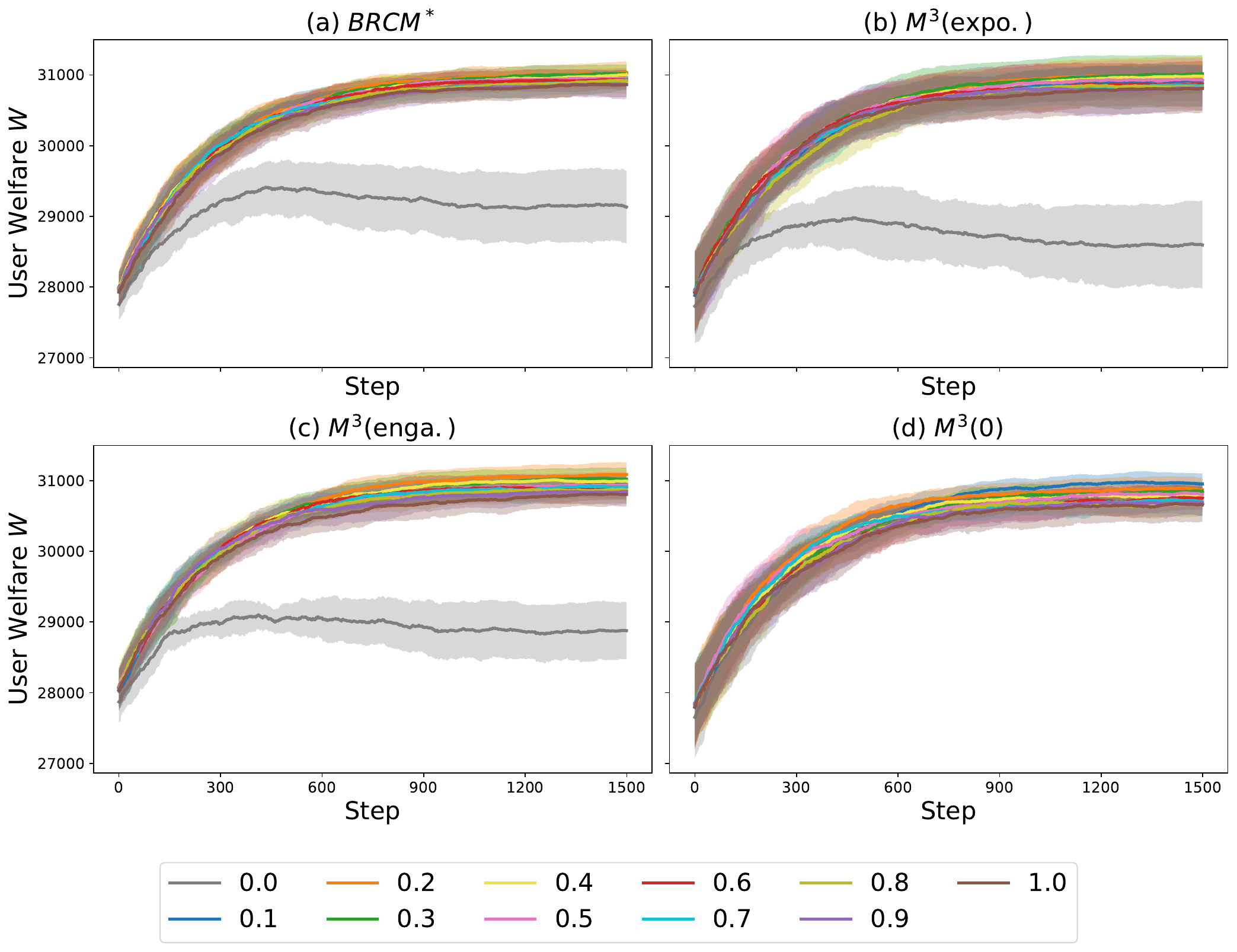}
        \caption{$K=5$.}
        \label{fig:amz-top5}
    \end{subfigure}
    \caption{User welfare evolution on the Instant-Video dataset with $K=1$ and $K=5$.}
    \label{fig:amz-combined}
\end{figure}

\end{document}